\documentclass{article}
\usepackage{graphicx}
\graphicspath{./}

\usepackage[utf8]{inputenc}
\usepackage[margin=1in]{geometry}
\usepackage{color}
\usepackage{authblk}

\usepackage{hyperref}
\usepackage{hypcap}

\usepackage[ruled,vlined]{algorithm2e}
\usepackage{algorithmic}

\usepackage{float}
\usepackage{caption}
\usepackage{subcaption}
\usepackage{setspace}
\usepackage{tcolorbox}

\usepackage{amsmath,amsfonts,amssymb,mathtools, amsthm}
\usepackage[english]{babel}
\newtheorem{prop}{Proposition}
\newtheorem{theorem}{Theorem}

\newtheorem{coro}{Corollary}
\newtheorem{lemma}{Lemma}
\newtheorem{remark}{Remark}

\numberwithin{equation}{section}

\title{A conservative Galerkin solver for the quasilinear diffusion model in magnetized plasmas}

\author[*]{Kun Huang}
\author[$\dag$]{Michael Abdelmalik}
\author[$\ddag$]{Boris Breizman}
\author[*]{Irene M. Gamba}

\affil[*]{Oden Institute for Computational Sciences and Engineering, University of Texas at Austin}
\affil[$\ddag$]{Institute for Fusion Studies, University of Texas at Austin}
\affil[$\dag$]{Department of Mechanical Engineering, Eindhoven University of Technology}

\date{}

\usepackage{subfiles} % Best loaded last in the preamble

\begin{document}

\maketitle

\begin{abstract}
The quasilinear theory describes the resonant interaction between particles and waves with two coupled equations: one for the evolution of the particle probability density function(\textit{pdf}), the other for the wave spectral energy density(\textit{sed}). In this paper, we propose a conservative Galerkin scheme for the quasilinear model in three-dimensional momentum space and three-dimensional spectral space, with cylindrical symmetry. 

    We construct an unconditionally conservative weak form, and propose a discretization that preserves the unconditional conservation property, by "unconditional" we mean that conservation is independent of the singular transition probability. The discrete operators, combined with a consistent quadrature rule, will preserve all the conservation laws rigorously. The technique we propose is quite general: it works for both relativistic and non-relativistic systems, for both magnetized and unmagnetized plasmas, and even for problems with time-dependent dispersion relations. 

    We represent the particle \textit{pdf} by continuous basis functions, and use discontinuous basis functions for the wave \textit{sed}, thus enabling the application of a positivity-preserving technique. The marching simplex algorithm, which was initially designed for computer graphics, is adopted for numerical integration on the resonance manifold. We introduce a semi-implicit time discretization, and discuss the stability condition. In addition, we present numerical examples with a "bump on tail" initial configuration, showing that the particle-wave interaction results in a strong anisotropic diffusion effect on the particle \textit{pdf}.

\vspace{6pt}\textbf{Keywords:} computational kinetic systems, quasilinear theory for magnetized plasmas, mean-field effect, weak turbulence model.
\end{abstract}

\section{Introduction} \label{intro}
The Vlasov-Maxwell system and the Vlasov-Poisson system are widely used to describe the collective(mean-field) effect of particles. Although a lot of work has been done in numerical methods for these systems, in practice a reduced model is often preferred when the problem is in high dimension and some loss of details is justified from physics consideration. 

For example, the electron runaway problem, which is the motivation of our research, has attracted a lot of attention as a risk factor for magnetic confinement fusion reactors like ITER \cite{rosenbluth1997theory, breizman_physics_2019}. Runaway electrons are a group of extremely fast electrons generated inside the tokamak, the release of which can damage the wall. Therefore, it is important to have an answer to the questions like how they are generated and how to mitigate them. 

The dynamics of runaway electrons is determined by external electromagnetic fields, collision, and particle-wave interaction. The quasilinear theory, as a reduced mean-field model governing particle-wave interaction, arises from averaging and linearizing over the original Vlasov-Maxwell system in weak turbulence regime. 

     The quasilinear theory for unmagnetized plasmas was proposed by Vedenov et al. \cite{vedenov1961nonlinear} and Drummond et al. \cite{drummond1962non}. It was later generalized by Shapiro et al. \cite{shapiro1962nonlinear} to model the magnetized plasma. The same idea has been used extensively in the following years, for example in the work of Kennel \cite{kennel_velocity_1966}, Lerche \cite{lerche1968quasilinear}, and Kaufman \cite{kaufman1971resonant}, etc. The validity of such a model reduction was studied numerically by Besse et al. \cite{besse2011validity}, and analytically by Bardos and Besse \cite{bardos2021diffusion} for Vlasov-Poisson system. The existence of global weak solutions in one dimensional electrostatic case has been proved in \cite{huang2023existence}.

Since the quasilinear theory studies the spectrum of waves and the averaged particle distribution function, it does not require a small time step to characterize the high wave frequency numerically.  However, the numerical computation of the particle-wave resonance system is still challenging, due to the resonance condition described with the Dirac delta function, the complicated dispersion relation, high dimension, nonlinearity, and conservation laws consisting of integrals in two different spaces. Therefore, although the theory has been widely used in physics, there is no preceding work focusing on the numerical method for quasilinear theory in magnetized plasmas. 

% Pokol et al.\cite{pokol_quasi-linear_2008, pokol2014quasi} did numerical simulation with the simplified quasilinear diffusion model for whistler waves and extraordinary electron waves. Aleynikov et al.\cite{aleynikov_stability_2015} used a ray-tracing code to study the wave growth and propagation with a given particle \textit{pdf}. 

% In astrophysics, Jeong et al.\cite{jeong2020quasi} used a first-order Taylor expansion for the dispersion relation and solved the coupling system with the finite difference method. The first work that numerically solved the whole quasilinear system with exact dispersion relation was done by Liu et al.\cite{liu_role_2018}. However, none of the above numerical solvers preserves conservation laws.

In this paper, we propose a conservative Galerkin solver for the homogeneous quasilinear particle-wave interaction system. 

Despite being a paradigm approach in the analysis and discretization of other kinetic equations, the weak formulation of the quasilinear model has not gained enough attention, partly because the equation for particles was usually written in a nonlinear diffusion form, and the equation for waves was treated as independent first-order ODEs with parameters. There are infinitely many equivalent forms to the same equation because of the resonance condition. Among all the equivalent forms, some are superior to the others, the reason is as follows.
    
    The quasilinear theory inherits the conservation laws from the original Vlasov-Maxwell system. However, generally the conservation is conditional, which means the gain and loss parts only offset each other on the resonance manifold. When the resonance manifold is broadened or approximated, conservation laws are no longer guaranteed. In this paper, we propose a novel integro-differential form and the corresponding unconditionally conservative weak form.
    
    It is desired that the discrete weak form will preserve the unconditional conservation property above, unfortunately, naive standard finite element discretizations turn out to fail. We located the cause of discretization errors by analyzing the weak form, and managed to construct a perfect discretization by replacing some quantities with their projection in the discrete finite element spaces.

    Apart from that, for numerical integration on resonance manifold, we adopt the marching simplex algorithm\cite{doi1991efficient, min2007geometric}, which enables us to deal with arbitrary wave modes.

\bigskip
This paper is organized as follows. In section \ref{pwsys} we review the relativistic quasilinear model for magnetized plasmas and introduce the integro-differential system with its weak form. The conservative semi-discrete system, as the main result of this paper, will be presented in section \ref{discrete}. In section \ref{tensor}, we derive the nonlinear ODE system associated with our conservative semi-discrete form, and the relation between two interaction tensors is proved. Stability and positivity will be discussed in section \ref{stability}. The numerical results are presented in section \ref{results}.

\section{The Quasilinear Particle-Wave Interaction System} \label{pwsys}

The quasilinear particle-wave interaction system consists of a diffusion equation for electron \textit{pdf}(probability distribution function) and a reaction equation for wave \textit{sed}(spectral energy density). They couple with each other through the coefficients. As a reduced model for the Vlasov-Maxwell system, the quasilinear theory inherits the conservation properties: mass, momentum, and energy. Moreover, the entropy of particles dissipates as a result of diffusion. 

In this section, we are going to show that the system can be written in a novel integro-differential form, which will finally lead to a conservative discrete scheme.

\subsection{The Integro-Differential Strong Form}
There are two ways to interpret the quasilinear theory. In classical language, it is a model reduction of the Vlasov-Maxwell system in a weak turbulence regime. Meanwhile, in quantum mechanical language, the waves in a plasma can also be regarded as a group of plasmons (wave packets, analogous to photons). Hence the interaction between particles and waves can be interpreted as a stochastic process of particles emitting/absorbing plasmons. For a derivation in quantum mechanical language, we refer the readers to the review paper of Vedenov\cite{vedenov1967theory} and the book of Thorne et al.\cite{thorne2017modern}. In what follows, we will rely more on the latter interpretation.

The particle \textit{pdf} is a non-negative function defined on the particle momentum space $\mathbb{R}_{p}^3$, i.e.
\begin{equation*}
	f(\mathbf{p},t): \mathbb{R}_{p}^3 \times \mathbb{R}_{+} \rightarrow \mathbb{R}_{+}. 
\end{equation*}

Each particle carries mass $m$, momentum $\mathbf{p}$, and energy $E(\mathbf{p})$. The particle velocity is $\mathbf{v}(\mathbf{p}) = \nabla_{\mathbf{p}}E(\mathbf{p})$.

At the same time, the wave \textit{sed} is a non-negative function defined on the wave momentum space $\mathbb{R}_{k}^3$ (the spectral space), i.e.
\begin{equation*}
	W(\mathbf{k},t) : \mathbb{R}_{k}^3 \times \mathbb{R}_{+} \rightarrow \mathbb{R}_{+}.
\end{equation*}

Each plasmon carries momentum
$\hbar\mathbf{k}$, energy $\hbar\omega$ and no mass. The relation between wave frequency $\omega$ and wave-vector
$\mathbf{k}$ is called the dispersion relation, $\omega:=\omega(\mathbf{k})$. 
The wave spectral energy density can be expressed in terms of plasmon
number density, $W(\mathbf{k})=N(\mathbf{k})\hbar\omega(\mathbf{k})$. Therefore it makes no difference whether we use $W(\mathbf{k})$ or $N(\mathbf{k})$ as the unknown.

The equations for quasilinear particle-wave interaction share the following structure,

    \begin{equation}
    \begin{split}
        \frac{\partial}{\partial t}f(\mathbf{p},t)&=\nabla_{p}\cdot\left(D[W](\mathbf{p},t)\cdot\nabla_{p}f(\mathbf{p},t)\right),\\
        \frac{\partial}{\partial t}W(\mathbf{k},t)&=\Gamma[f](\mathbf{k},t)W(\mathbf{k},t).
    \end{split}
    \label{diffreac}
\end{equation}

% \begin{equation}
%     \begin{split}
%         \frac{\partial}{\partial t}f(\mathbf{p},t)&=div_{{\bf p}}\left(D[W](\mathbf{p},t)\cdot\nabla_{\mathbf{p}}f(\mathbf{p},t)+ A(\mathbf{p},t)f(\mathbf{p},t)\right)\\
%         \frac{\partial}{\partial t}W(\mathbf{k},t)&=\Gamma[f](\mathbf{k},t)W(\mathbf{k},t)+R[f](\mathbf{k},t),
%     \end{split}
%     \label{diffreac}
% \end{equation}
Both relations, $D\!:\! L^{1}(\mathbb{R}_{k}^3) \rightarrow (L^\infty(\mathbb{R}_{p}^3))^{3\times3}$ and $\Gamma\!: \!H^{1}(\mathbb{R}_{p}^3) \!\rightarrow \! L^\infty(\mathbb{R}_{k}^3)$, are determined by transition probabilities of the stochastic emission/absorption process. The transition probabilities per se, depend solely on pre-interaction and post-interaction kinetic variables: particle momentum $\mathbf{p}$, particle energy $E(\mathbf{p})$, plasmon momentum $\hbar\mathbf{k}$ and plasmon energy $\hbar\omega(\mathbf{k})$. Hence the particle energy relation $E(\mathbf{p})$ and plasmon dispersion relation $\omega(\mathbf{k})$ must be specified before numerical simulation.

\begin{remark}
The wave dispersion relation $\omega(\mathbf{k})$ depends on the medium, i.e. the plasma itself, which is evolving. Since the computational cost for an accurate dispersion relation $\omega:=\omega(\mathbf{k})$ can be quite high, there is, in
practice, a tendency to use low-order approximations based on appropriate assumptions, for example, the cold plasma assumption (see Appendix).
\end{remark}

\begin{remark}
    In a plasma, there can be multiple wave modes, i.e. multiple "species" of plasmons, each having a distinct dispersion relation $\omega(\mathbf{k})$. In our numerical experiment, we use the dispersion relation of whistler waves in a cold magnetized plasma. Nevertheless, our numerical method is compatible with any dispersion relation, and can be used to simulate multiple wave modes at the same time.
\end{remark}

In the previous paragraphs, we introduced the general governing equations for particles absorbing/emitting plasmons. Next, let us focus on a specific example that the proposed method is designed for.

\subsubsection*{The Particle-Wave Interaction System for Magnetized Plasmas with Cylindrical Symmetry}
We are interested in particle-wave interaction for plasmas embedded in the magnetic field $\mathbf{B}(t)=\mathbf{B}_0$, as a consequence, we will focus on 
gyro-averaged particle distribution functions $f(\mathbf{p},t)$ having cylindrical symmetry. For simplicity, it is further assumed that $W(\mathbf{k},t)$ is cylindrically symmetric. 

The background magnetic field $\mathbf{B}_{0}$ induces an axis direction $\mathbf{e}_{\parallel}=\mathbf{B}_0/|\mathbf{B}_0|$ and the associated cylindrical coordinates, where $u_{\parallel}=\mathbf{u}\cdot\mathbf{e}_{\parallel}$ and $u_{\perp}=({u^2-u_{\parallel}^2})^{1/2}$ for any vector $\mathbf{u}\in\mathbb{R}^3$. Then we have $f(\mathbf{p},t) = f(p_{\parallel}, p_{\perp},t)$ and $W(\mathbf{k},t) = W(k_{\parallel}, k_{\perp},t)$.

Particles with charge $q$ and mass $m$ has gyro-frequency $\omega_{c}=qB_{0}/mc$. Note that for electrons, $q=-e$. Particles with momentum $\mathbf{p}$ has relativistic energy $E(\mathbf{p})=\gamma(\mathbf{p})mc^{2}$, where $\gamma$ is the Lorentz factor.

\bigskip
To model the runaway electrons in a tokamak, Breizman et al.\cite{breizman_physics_2019} used the following equations written in spherical coordinates(i.e. $\theta = \arccos{p_{\parallel}/p}$),
\begin{equation}
    \begin{split}
        \frac{\partial f(\mathbf{p},t)}{\partial t} &= \frac{1}{p^{2}}\frac{\partial}{\partial p}p^{2}(D_{pp}\frac{\partial f}{\partial p}+D_{p\theta}\frac{\partial f}{p\partial\theta})+\frac{1}{p\sin\theta}\frac{\partial}{\partial\theta}\sin\theta(D_{\theta p}\frac{\partial f}{\partial p}+D_{\theta\theta}\frac{\partial f}{p\partial\theta}),\\
        \frac{\partial W(\mathbf{k},t)}{\partial t}&=\Gamma(\mathbf{k},t)W,
    \end{split}
    \label{mageq}
\end{equation}
where the diffusion coefficients are weighted integrals of wave \textit{sed} $W(\mathbf{k},t)$,
\begin{equation*}
    \begin{split}
        D_{pp}(\mathbf{p},t)&=\sum_{l=-\infty}^{+\infty}\int d^{3}k\left\{ W(\mathbf{k},t)U_{l}(\mathbf{p},\mathbf{k})\delta(\omega(\mathbf{k})-k_{\parallel}v(\mathbf{p})\cos\theta-l\omega_{c}/\gamma)\right\},\\
        D_{p\theta}(\mathbf{p},t)=D_{p\theta}(\mathbf{p},t)&=\sum_{l=-\infty}^{+\infty}\int d^{3}k\left\{\left(\frac{l\omega_{c}/\gamma-\omega\sin^{2}\theta}{\omega\sin\theta\cos\theta}\right)W(\mathbf{k},t)U_{l}(\mathbf{p},\mathbf{k})\delta(\omega-k_{\parallel}v\cos\theta-l\omega_{c}/\gamma)\right\},\\
        D_{\theta\theta}(\mathbf{p},t)&=\sum_{l=-\infty}^{+\infty}\int d^{3}k\left\{\left(\frac{l\omega_{c}/\gamma-\omega\sin^{2}\theta}{\omega\sin\theta\cos\theta}\right)^{2}W(\mathbf{k},t)U_{l}(\mathbf{p},\mathbf{k})\delta(\omega-k_{\parallel}v\cos\theta-l\omega_{c}/\gamma)\right\},
    \end{split}
\end{equation*}
and the growth rate is a weighted integral of particle \textit{pdf}'s gradient, $\nabla_{\mathbf{p}}f(\mathbf{p},t)$,
\begin{equation*}
    \Gamma(\mathbf{k},t)=\sum_{l=-\infty}^{+\infty}\int d^{3}p\left\{ v(\mathbf{p})(\frac{\partial f}{\partial p}+\frac{l\omega_{c}/\gamma-\omega\sin^{2}\theta}{\omega(\bf k)\sin\theta\cos\theta}\frac{\partial f}{p\partial\theta})U_{l}(\mathbf{p},\mathbf{k})\delta(\omega-k_{\parallel}v\cos\theta-l\omega_{c}/\gamma)\right\}.
\end{equation*}

The coefficients characterizing transition probability, is a function of $\mathbf{p}$ and $\mathbf{k}$ that always takes finite non-negative values,
\begin{equation}
    U_{l}(\mathbf{p};\mathbf{k})=8\pi^{2}e^{2}\frac{\left\{ \frac{l\omega_{c}}{k_{\perp}p}J_{l}+E_{3}\cos\theta J_{l}+iE_{2}\sin\theta J_{l}^{\prime}\right\} ^{2}}{\left(1-E_{2}^{2}\right)\frac{1}{\omega}\frac{\partial}{\partial\omega}(\omega^{2}\varepsilon)+2iE_{2}\frac{1}{\omega}\frac{\partial}{\partial\omega}(\omega^{2}g)+E_{3}^{2}\frac{1}{\omega}\frac{\partial}{\partial\omega}(\omega^{2}\eta)},
    \label{U_l}
\end{equation}
where the dielectric tensor components $\varepsilon, g $, $\eta$ and the wave polarization vector components $E_{j}$ are functions of wave frequency $\omega$ as defined in Equation(\ref{dielec_compo}) and Equation(\ref{polarize}), argument of the Bessel functions is $k_{\perp}p_{\perp}/m\omega_{c}$, for details, see Breizman et. al.\cite{breizman_physics_2019}.

The above formulas are sufficient to perform a trivial numerical simulation: treat Equation(\ref{diffreac}) as a normal diffusion equation and a normal reaction equation with time-varying coefficients. The challenging part is to preserve conservation, especially when there are integrals containing the Dirac delta function. Numerical integrals are always performed by quadrature rules, however, the quadrature points usually do not reside exactly on the resonance manifold. In this paper, we propose an unconditionally conservative approach by employing a novel equivalent integro-differential form instead of the original equation. By "unconditional" we mean that the scheme is conservative no matter how the resonance manifold is discretized or broadened.

\subsubsection*{The Emission/Absorption Kernel and Directional Differential Operator}
To rephrase Equation(\ref{mageq}) in integral-differential form, let us introduce two important concepts, the emission/absorption kernel and the directional differential operator, along with some necessary notations.

Particles with momentum $\mathbf{p}$ do not emit or absorb plasmons with wave vector $\mathbf{k}$ unless a certain resonance condition is satisfied. Define the $l$-th resonance indicator function $s_{l}(\mathbf{p},\mathbf{k})\coloneqq \omega(\mathbf{k})-k_{\parallel}v_{\parallel}-l\omega_{c}/\gamma(\mathbf{p})$, $l\in \mathbb{Z}$, then the resonance condition reads $s_{l}(\mathbf{p},\mathbf{k})=0$. Define the $l$-th resonance manifold as

\begin{equation}
    \mathcal{S}_{l} = \left\{(\mathbf{p},\mathbf{k})\in \mathbb{R}_{p}^3 \times \mathbb{R}_{k}^3: s_{l}(\mathbf{p},\mathbf{k})=\omega(\mathbf{k})-k_{\parallel}v_{\parallel}-l\omega_{c}/\gamma(\mathbf{p})=0 \right\},
    \label{resomani}
\end{equation}

then we can say that particles with momentum $\mathbf{p}$ emit or absorb plasmons with wave vector $\mathbf{k}$ only when $(\mathbf{p},\mathbf{k})$ belongs to one of the resonance manifolds.

\bigskip
Analogous to the definition of collisional kernels in Boltzmann equations and Fokker-Planck-Landau equations, we define the \textit{emission/absorption kernel} which characterizes the probability for a particle with momentum $\mathbf{p}$ to absorb or emit a plasmon with wave vector $\mathbf{k}$,
\begin{equation}
    \mathcal{B}(\mathbf{p},\mathbf{k})=\sum_{l=-\infty}^{+\infty}U_{l}(\mathbf{p};\mathbf{k})\delta(s_{l}(\mathbf{p};\mathbf{k})).
    \label{kernel}
\end{equation}

As we have mentioned above, interaction happens only if the resonance condition is satisfied, so the emission/absorption kernel contains a Dirac delta function. The coefficients $U_l(\mathbf{p},\mathbf{k})$ are given in Equation(\ref{U_l}). They take finite non-negative values for any coordinates $(\mathbf{p},\mathbf{k})$.

\bigskip
Interaction with a plasmon results in diffusion of particle \textit{pdf} $f(\mathbf{p},t)$ along a particular direction $\mathbf{\beta}(\mathbf{p},\mathbf{k})$, thus we define the directional differential operator

\begin{equation}
    \mathcal{L} (g(\mathbf{p})) \coloneqq \mathbf{\beta}(\mathbf{p},\mathbf{k})\cdot \nabla_{p} g(\mathbf{p}) = \frac{k_{\parallel}v_{\parallel}}{\omega}\frac{p}{p_{\parallel}}\frac{\partial g}{\partial p_{\parallel}}+(1-\frac{k_{\parallel}v_{\parallel}}{\omega})\frac{p}{p_{\perp}}\frac{\partial g}{\partial p_{\perp}}
    \label{direcderiv}
\end{equation}

Further, define the $L^2$ inner product in particle momentum space $(u(\mathbf{p}),v(\mathbf{p}))_{p}\coloneqq\int_{\mathbb{R}_{p}^3}uvd^{3}p$, and the $L^2$ inner product in wave spectral space $(U(\mathbf{k}),V(\mathbf{k}))_{k}=\int_{\mathbb{R}_{k}^3}UVd^{3}k$. 

Denote the adjoint operator of $\mathcal{L}$ by $\mathcal{L}^{*}$, then by definition, we have
\begin{equation*}
    (\mathcal{L}^{*}u,v)_{p}=(u,\mathcal{L}v)_{p}.
\end{equation*}

\subsubsection*{The Bilinear Integro-Differential Operators}
Now all the ingredients are prepared, we claim that the diffusion term and reaction term can be rewritten as bilinear integro-differential operators.

\begin{theorem}
The particle-wave interaction system in Equation(\ref{mageq}) is equivalent to
\begin{equation}
    \begin{split}
        \frac{\partial f(\mathbf{p},t)}{\partial t}&=\mathbb{B}(W,f)=-\int_{\mathbb{R}_{k}^3}\mathcal{L}^{*}\left(\mathcal{B}(\mathbf{p},\mathbf{k})W\mathcal{L}f\right)d\mathbf{k},\\
        \frac{\partial W(\mathbf{k},t)}{\partial t}&=\mathbb{H}(f,W)=\int_{\mathbb{R}_{p}^3}(\mathcal{L}E(\mathbf{p}))W\mathcal{B}(\mathcal{L}f)d\mathbf{p}.
    \end{split}
    \label{abstr}
\end{equation}
\end{theorem}

\begin{proof}
By definition,
\begin{equation*}
    \int\delta(s(\mathbf{x}))s(\mathbf{x})g(\mathbf{x})d\mathbf{x}=\frac{1}{\vert\nabla_{x}s\vert}\int_{\{\mathbf{x}:s(\mathbf{x})=0\}}s(\mathbf{x})g(\mathbf{x})d\pi_{s}=0
\end{equation*}

Note that
\begin{equation*}
        \frac{l\omega_{c}/\gamma-\omega\sin^{2}\theta}{\omega\sin\theta\cos\theta}=\frac{\omega\cos\theta-k_{\parallel}v}{\omega\sin\theta},
\end{equation*}
on the resonance manifold.

We obtain Equation(\ref{abstr}) by substituting the above identity into Equation(\ref{mageq}) and rewriting everything in cylindrical coordinates(i.e. replacing $\theta$ with $\arccos{p_{\parallel}/p}$).
\end{proof}

\begin{remark}
 Both the particle diffusion operator $\mathbb{B}$ and the wave reaction operator $\mathbb{H}$ mix particle momentum $\mathbf{p}$ and plasmon wave vector $\mathbf{k}$ through the absorption/emission kernel $\mathcal{B}(\mathbf{p},\mathbf{k})$.
\end{remark}

\begin{remark}
One might have noticed that $\mathcal{L} E(\mathbf{p}) = v(\mathbf{p})$. The reason we write $\mathcal{L} E$ rather than $v$ is to induce our conservative semi-discrete form and to save preprocessing time. The details will be addressed next.
\end{remark}

\subsection{The Unconditionally Conservative Weak Form and H-Theorem}
 For the purpose of either modeling or numerical implementation, the
absorption/emission kernel $\mathcal{B}(\mathbf{p},\mathbf{k})$ is
usually replaced with its approximation $\mathcal{B}_{\varepsilon}(\mathbf{p},\mathbf{k})$.
    Here we present two examples for such approximation.
    
    \begin{itemize}
        \item Approximation to the identity.
        
        The kernel is approximated with
            \begin{equation}
                \mathcal{B}_{h}^{\text{ati}}(\mathbf{p},\mathbf{k})=\sum_{l}U_{l}(\mathbf{p},\mathbf{k})\frac{1}{\varepsilon}\varphi\left(\frac{s_{l}(\mathbf{p},\mathbf{k})}{\varepsilon}\right),
                \label{ATI}
            \end{equation}
        where the compactly supported and positive function $\varphi(z)$ has unit mass, i.e. $\int_{\mathbb{R}}\varphi(z)dz=1$.  

        Recall the definition of resonance manifold in Equation(\ref{resomani}), it is a hypersurface implicitly determined by the resonance condition. The above approximation is equivalent to broadening of the resonance manifold, as the approximated hypersurface has finite “width” proportional to $\varepsilon$.  
        \item Marching cube/simplex algorithm. 
        
        The kernel is approximated with
        \begin{equation}
            \mathcal{B}_{\varepsilon}^{\text{msa}}(\mathbf{p},\mathbf{k})=\sum_{l}U_{l}(\mathbf{p},\mathbf{k})\delta\left(L_{\varepsilon}s_{l}(\mathbf{p},\mathbf{k})\right),
            \label{MSA}
        \end{equation}
        where $L_{\varepsilon}s_{l}$ represents the piecewise linear interpolation
        of $s_{l}$. As have been illustrated in \cite{doi1991efficient, min2007geometric}, such approximation discretizes the resonance manifold, i.e. replaces the smooth hypersurface with a disjoint union of simplices, thus enabling convenient numerical integration.
    \end{itemize}
    
    In what follows we derive the special weak form, and prove that even if the emission/absorption kernel $\mathcal{B}$ is replaced/approximated, we can still preserve mass, momentum and energy with the proposed form.

 To obtain the weak formulation associated with the system(\ref{abstr}), test it with $\phi(\mathbf{p})$ and $\eta(\mathbf{k})$, we obtain that
\begin{equation*}
    \begin{split}
        \int_{\mathbb{R}_{p}^3}\frac{\partial f}{\partial t}\phi d\mathbf{p}&=-\int_{\mathbb{R}_{k}^3}d\mathbf{p}\int_{\mathbb{R}_{p}^3}d\mathbf{k}\{\mathcal{L}\phi\mathcal{L}fW\mathcal{B}\}\\
        \int_{\mathbb{R}_{k}^3}\frac{\partial W}{\partial t}\eta d\mathbf{k}&=\int_{\mathbb{R}_{k}^3}d\mathbf{k}\int_{\mathbb{R}_{p}^3}d\mathbf{p}\{\eta\mathcal{L}E\mathcal{L}fW\mathcal{B}\}
    \end{split}
	\label{weakform}
\end{equation*}

Note that the order of integration here is different on the right-hand side. In what follows, assume that $\int_{\mathbb{R}_{k}^3\times\mathbb{R}_{p}^3}d\mathbf{k}d\mathbf{p}|\mathcal{L}\phi\mathcal{L}fW\mathcal{B}|$ and $\int_{\mathbb{R}_{k}^3\times\mathbb{R}_{p}^3}d\mathbf{k}d\mathbf{p}|\eta\mathcal{L}E\mathcal{L}fW\mathcal{B}|$ are finite, therefore by Fubini's theorem, the order of integration does not matter,

\begin{equation*}
    \begin{split}
        \int_{\mathbb{R}_{p}^3}\left(\int_{\mathbb{R}_{k}^3}\mathcal{L}\phi\mathcal{L}fW\mathcal{B}d\mathbf{k}\right)d\mathbf{p}=\int_{\mathbb{R}_{k}^3}\left(\int_{\mathbb{R}_{p}^3}\mathcal{L}\phi\mathcal{L}fW\mathcal{B}d\mathbf{p}\right)d\mathbf{k}=\iint_{\mathbb{R}_{p}^3\times\mathbb{R}_{k}^3}d\mathbf{k}d\mathbf{p}\left\{ \mathcal{L}\phi\mathcal{L}fW\mathcal{B}\right\} \\
        \int_{\mathbb{R}_{k}^3}\left(\int_{\mathbb{R}_{p}^3}\eta \mathcal{L} E\mathcal{L}fW\mathcal{B}d\mathbf{p}\right)d\mathbf{k}=\int_{\mathbb{R}_{p}^3}\left(\int_{\mathbb{R}_{k}^3}\eta \mathcal{L} E\mathcal{L}fW\mathcal{B}d\mathbf{k}\right)d\mathbf{p}=\iint_{\mathbb{R}_{p}^3\times\mathbb{R}_{k}^3}d\mathbf{k}d\mathbf{p}\left\{ \eta \mathcal{L} E\mathcal{L}fW\mathcal{B}\right\}.
    \end{split}
\end{equation*}

On the right-hand side is inner products of bilinear integro-differential operators with test functions. Therefore to simplify the notation, we can define trilinear forms $B$ and $H$ as follows:
\begin{equation}
    \begin{split}
        B(f,W,\phi)&\coloneqq\iint_{\mathbb{R}_{p}^3\times\mathbb{R}_{k}^3}d\mathbf{k}d\mathbf{p}\left\{ \mathcal{L}\phi\mathcal{L}fW\mathcal{B}\right\},\\
        H(W,f,\eta)&\coloneqq\iint_{\mathbb{R}_{p}^3\times\mathbb{R}_{k}^3}d\mathbf{k}d\mathbf{p}\left\{ \eta \mathcal{L} E\mathcal{L}fW\mathcal{B}\right\}.
    \end{split}
    \label{trilin}
\end{equation}

As a result, the weak form of system(\ref{abstr}) can be written as, 
\begin{equation}
    \begin{split}
        \left(\frac{\partial f}{\partial t},\phi\right)_{p}&=-B(f,W,\phi),\\
        \left(\frac{\partial W}{\partial t},\eta\right)_{k}&=H(W,f,\eta).
    \end{split}
    \label{simpweakform}
\end{equation}

Due to resonance, there are infinitely many equivalent forms for the same equation, for example,
    \begin{equation*}
        \sum_{l}U_{l}(\mathbf{p},\mathbf{k})\delta(\omega-k_{\parallel}v_{\parallel}-l\omega_{c}/\gamma)\left(\frac{k_{\parallel}v_{\parallel}}{\omega}\frac{p}{p_{\parallel}}\frac{\partial\phi}{\partial p_{\parallel}}+(1-\frac{k_{\parallel}v_{\parallel}}{\omega})\frac{p}{p_{\perp}}\frac{\partial\phi}{\partial p_{\perp}}\right)
    \end{equation*}
    is always equal to
    \begin{equation*}
        \sum_{l}U_{l}(\mathbf{p},\mathbf{k})\delta(\omega-k_{\parallel}v_{\parallel}-l\omega_{c}/\gamma)\left(\left(\frac{k_{\parallel}v_{\parallel}+l\omega_{c}/\gamma}{\omega}\right)^{\alpha}\frac{k_{\parallel}v_{\parallel}}{\omega}\frac{p}{p_{\parallel}}\frac{\partial\phi}{\partial p_{\parallel}}+(1-\frac{k_{\parallel}v_{\parallel}}{\omega})\frac{p}{p_{\perp}}\frac{\partial\phi}{\partial p_{\perp}}\right),
    \end{equation*}
    for any constant $\alpha>0$. 
    
    In the following theorem, we prove the superiority of the proposed form, i.e. the unconditional conservation property.

\begin{theorem}[unconditional conservation]\label{uncond}

If $f(\mathbf{p},t)$ and $W(\mathbf{k},t)$ solve the system(\ref{simpweakform}) with emission/absorption kernel being replaced by $\mathcal{B}_{\varepsilon}$, then for any $\mathcal{B}_{\varepsilon}$ we have the following conservation laws,

\begin{itemize}
    \item Mass Conservation
    \begin{equation*}
        \frac{\partial}{\partial t}\mathcal{M}_{tot}=\frac{\partial}{\partial t}\left((f,1)_{p}+(W,0)_{k}\right)=0
    \end{equation*}
    
    \item Momentum Conservation
    \begin{equation*}
            \frac{\partial}{\partial t}\mathcal{P}_{\parallel}^{tot}
            =\frac{\partial}{\partial t}\left((f,p_{\parallel})_{p}+(\frac{W}{\hbar \omega},\hbar k_{\parallel})_{k}\right)=0,
    \end{equation*}
    \item Energy Conservation
    \begin{equation*}
        \frac{\partial}{\partial t} \mathcal{E}_{tot}=\frac{\partial}{\partial t}((f,E)_{p}+(\frac{W}{\hbar \omega},\hbar \omega)_{k})=0
    \end{equation*}
    
\end{itemize}
\end{theorem}

\begin{proof}
Sum the two rows in Equation(\ref{simpweakform}), we have
\begin{equation}
	\left(\frac{\partial f}{\partial t},\phi\right)_{p}+\left(\frac{\partial W}{\partial t},\eta\right)_{k}=\int_{\mathbb{R}_{p}^3\times\mathbb{R}_{k}^3}d\mathbf{k}d\mathbf{p}\left\{ \left(\eta\mathcal{L}E-\mathcal{L}\phi\right)\mathcal{L}fW\mathcal{B}_{\varepsilon}\right\} 
	\label{cons_eq}
\end{equation}
Substitute the test functions $\{\phi, \eta\}$ with conservation quantity pairs $\{\phi_c, \eta_c\}$, i.e. mass pair $\{1, 0\}$, parallel momentum pair $\{p_{\parallel}, \frac{k_{\parallel}}{\omega}\}$ and energy pair $\{E, 1\}$. In principle the condition for conservation is $\eta_{c}\mathcal{L}E-\mathcal{L}\phi_{c}=0$ on the approximate resonance manifold $\mathcal{S}=\left\{ (\mathbf{p},\mathbf{k})\in\mathbb{R}_{p}^{3}\times\mathbb{R}_{k}^{3}:\mathcal{B}_{\varepsilon}(\mathbf{p},\mathbf{k})\neq0\right\}$. However, due to our particular definition of the directional differential operator $\mathcal{L}$, we actually have $\eta_{c}\mathcal{L}E-\mathcal{L}\phi_{c}=0$ on the whole domain $\mathbb{R}_{p}^3\times  \mathbb{R}_{k}^3$. Therefore the conservation laws hold regardless of emission/absorption kernel $\mathcal{B}_{\varepsilon}$.
\end{proof}

\begin{remark}
The unconditionally conservative form also exists for unmagnetized plasmas with cylindrical symmetry, where we just replace the emission/absorption kernel with $\mathcal{B}=\frac{1}{2\pi}\int_{0}^{2\pi}\frac{\omega^{2}}{k^{2}v^{2}}\delta(\omega-k_{\parallel}v_{\parallel}-k_{\perp}v_{\perp}\cos\alpha)d\alpha$.
\end{remark}

\begin{remark}
        Since the conservation laws solely depend on $\eta\mathcal{L}E-\mathcal{L}\phi=0$, the unconditional conservative form and the scheme we are going to propose can be generalized for time-dependent dispersion relation $\omega(\mathbf{k};t)$ with no extra effort. An important example is the self-consistent dispersion relation $\omega = \omega(\mathbf{k}; f_{1}(t), f_{2}(t), \cdots)$. The only obstacle is the extra computational cost of updating the interaction tensors in each step. As will be shown in section \ref{tensor}, that calculation can be expensive.
    \end{remark}

Recall the definition of emission/absorption kernel $\mathcal{B}$, 
\begin{equation*}
	\mathcal{B}(\mathbf{p},\mathbf{k})=\sum_{l=-\infty}^{+\infty}U_{l}(\mathbf{p};\mathbf{k})\delta(s_{l}(\mathbf{p};\mathbf{k})).
\end{equation*}

Test the equation for particle \textit{pdf} with $\phi(\mathbf{p})=\log f(\mathbf{p})$, since $U_{l}$ and $W$ are non-negative, the right-hand side will be non-positive, 
\begin{equation*}
    (\frac{\partial f}{\partial t},\log f)_{p}=-\iint d\mathbf{k}d\mathbf{p}\frac{1}{f}(\mathcal{L}f)^{2}W\mathcal{B}\leq0,
\end{equation*}
thus we can prove the dissipation of entropy, i.e. H-theorem for the particle \textit{pdf},
\begin{equation}
    \frac{\partial}{\partial t}(f,\log f)_{p}=(\frac{\partial f}{\partial t},\log f)_{p}+(\frac{\partial f}{\partial t},1)_{p}\leq0.
\end{equation}

\section{The Conservative Discretization} \label{discrete}
This section aims to find a semi-discrete problem that consistently approximates the original system, and at the same time preserves discrete conservation laws. So in the following subsections, we will first introduce our finite element discretization, the necessary projection operators, and then elaborate on the conservation technique. 

\subsection{The Finite Element Discretization}

\subsubsection*{The Cut-Off Domain and Boundary Conditions}
Analogous to existing work on kinetic equations, for example, the papers of Zhang et al.\cite{zhang2017conservative, zhang2018conservative},  we assume that given any $0<\epsilon_p \ll 1$ and $0 < \epsilon_k \ll 1$, there exists finite cylindrical domains $\Omega_{p}^{L}\subsetneq\mathbb{R}_{p}^{3}$
and $\Omega_{k}^{L}\subsetneq\mathbb{R}_{k}^{3}$ such that  for any $t\geq0$,
\begin{equation*}
    \left|1-\frac{\int_{\Omega^{L}_{p}}f(\mathbf{p},t)d^{3}p}{\int_{\mathbb{R}_{p}^{3}}f(\mathbf{p},t)d^{3}p}\right|\leq\epsilon_{p},
\end{equation*}
and
\begin{equation*}
    \left|1-\frac{\int_{\Omega^{L}_{k}}W(\mathbf{k},t)d^{3}k}{\int_{\mathbb{R}_{k}^{3}}W(\mathbf{k},t)d^{3}k}\right|\leq\epsilon_{k}.
\end{equation*}

The particle momentum cut-off domain $\Omega_{p}^{L}$ is supposed to be adaptive, while in our numerical experiments it turns out that, as a result of anisotropic diffusion, there is no need to extend it.

Then it is reasonable to solve the equations in cut-off domains $\Omega^{L}_{p}$ and $\Omega^{L}_{k}$. For the wave \textit{sed} $W(\mathbf{k})$, there is no need for a boundary condition since there is no flux in wave vector space. For the particle \textit{pdf}, we have the following choices, and when the domain $\Omega^{L}_{p}$ is large enough, they are actually equivalent.

On the boundary $\partial\Omega^{L}_{p}$ of cut-off domain $\Omega^{L}_{p}$,
$|f|$ and $|\nabla_{\mathbf{p}}f|$ are nearly zero, two types of boundary
conditions can be applied,

\begin{enumerate}
	\item The zero-value boundary condition
		\begin{equation*}
		    f=0,\ \forall\mathbf{p}\in\partial\Omega^{L}_{p}
		\end{equation*}

	\item The zero-flux boundary condition
		\begin{equation*}
		    \left(D[W]\nabla_{\mathbf{p}}f\right)\cdot\mathbf{n}=0,\ \forall\mathbf{p}\in\partial\Omega^{L}_{p}
		\end{equation*}
\end{enumerate}

Suppose we test the diffusion equation with $\phi_{h} \in V_{h}$. With Neumann's boundary condition, i.e. in the zero-flux case, the semi-discrete weak form reads:
\begin{equation*}
    (\frac{\partial f_{h}}{\partial t},\phi_{h})+(D[W_{h}]\nabla_{\mathbf{p}}f_{h},\nabla_{\mathbf{p}}\phi_{h})=0
\end{equation*}

For Dirichlet's boundary conditions given by to zero-value on the discretized boundary, i.e. $f_{\text{ini}}\mid_{\partial\Omega^L_p} \equiv 0$,   Nitsche's method \cite{nitsche1971variationsprinzip} applies, hence the weak  the semi-discrete form reads
\begin{equation*}
    (\frac{\partial f_{h}}{\partial t},\phi_{h})+(D[W_{h}]\nabla_{\mathbf{p}}f_{h},\nabla_{\mathbf{p}}\phi_{h})-\langle\left(D[W_{h}]\nabla_{\mathbf{p}}f_{h}\right)\cdot\mathbf{n}_{p},\phi_{h}\rangle_{\partial \Omega^{L}_{p}}+\langle\left(D[W_{h}]\nabla_{\mathbf{p}}\phi_{h}\right)\cdot\mathbf{n}_{p},f_{h}\rangle_{\partial \Omega^{L}_{p}}=0
\end{equation*}

The only difference between them is the boundary integral, which can be below machine epsilon for large enough $\Omega^{L}_{p}$, because $D$ and $\phi_h$ are finite, while $|f_h|$ and $|\nabla_{\mathbf{p}}f_h|$ goes to zero as we enlarge the domain. Stability can be proved for both formulations, in the rest of the article, for simplicity, we will use the zero-flux boundary condition. 

\subsubsection*{The Finite Element Spaces}
Since we have assumed cylindrical symmetry, the 3P-3K problem actually becomes 2P-2K.
\begin{equation*}
    \begin{split}
        &\mathbf{p}=(p_{1},p_{2},p_{3})\in\Omega^{L}_{p}\Leftrightarrow(p_{\parallel},p_{\perp})\in\tilde{\Omega}^{L}_{p}\subset\mathbb{R}\times\mathbb{R}^{+}\\
        &\mathbf{k}=(k_{1},k_{2},k_{3})\in\Omega^{L}_{k}\Leftrightarrow(k_{\parallel},k_{\perp})\in\tilde{\Omega}^{L}_{k}\subset\mathbb{R}\times\mathbb{R}^{+}
    \end{split}
\end{equation*}

Let $\mathcal{T}_{h}^{p}=\{R_{p}\}$, $\mathcal{T}_{h}^{k}=\{R_{k}\}$
be rectangular partitions of $\tilde{\Omega}^{L}_{p}$ and
$\tilde{\Omega}^{L}_{k}$ respectively. We define the meshsize for momentum space as  $h_{p}=\textnormal{max}_{R_{p}\in\mathcal{T}_{h}^{p}}\textnormal{diam}(R_{p})$ and the meshsize for wave vector space as $h_{k}=\textnormal{max}_{R_{k}\in\mathcal{T}_{h}^{k}}\textnormal{diam}(R_{k})$. 

The test space for particle \textit{pdf} consists of continuous piecewise polynomials with degree $\alpha_{1}$,

\begin{equation}
    \mathcal{G}_{h}^{\alpha_1}=\{f(p_{\parallel},p_{\perp}) \in C^0(\Omega_p):f|_{R_{p}}\in Q^{\alpha_1}(R_{p}),\forall~R_{p}\in\mathcal{T}_{h}^{p}\}.
\end{equation}

The test space for wave \textit{sed} consists of discontinuous piecewise polynomials with degree $\alpha_{2}$,

\begin{equation}
    \mathcal{W}_{h}^{\alpha_2}=\{W(k_{\parallel},k_{\perp}):W|_{R_{k}}\in Q^{\alpha_2}(R_{k}),\forall~R_{k}\in\mathcal{T}_{h}^{k}\}.
\end{equation}
To ensure positivity of $W_{h}$, it is required that $\alpha_{2} = 0$ or $\alpha_{2} = 1$, the reason will be addressed later.

As will be shown in the next section, one of the key points to conservation is replacing $v_{\parallel}$, $v_{\perp}$ and $k_{\parallel}/\omega$ with $\frac{\partial E_{h}(\mathbf{p})}{\partial p_{\parallel}}$, $\frac{\partial E_{h}(\mathbf{p})}{\partial p_{\perp}}$ and $N_{\parallel,h}$, where $E_{h} = \Pi_{p,h} E(\mathbf{p})$ is the discrete particle kinetic energy, and $N_{\parallel,h}=\Pi_{k,h}N_{\parallel}$ is the discrete refraction index. The projection operators can be arbitrarily chosen as long as they satisfy the following conditions:
    \begin{enumerate}
        \item The projection $\Pi_{p,h}$ into test space $\mathcal{G}_{h}^{\alpha_1}$ must satisfy that 
        \begin{equation*}
            \ensuremath{\lim_{h\rightarrow0}\Vert\Pi_{p,h}g(\mathbf{p})-g(\mathbf{p})\Vert_{L^{2}(\Omega_{p}^{L})}=0},\ \forall g \in L^{2}(\Omega_{p}^{L}),
        \end{equation*}
        and 
        \begin{equation*}
            \ensuremath{\lim_{h\rightarrow0}\Vert\Pi_{p,h}E(\mathbf{p})-E(\mathbf{p})\Vert_{H^{1}(\Omega_{p}^{L})}=0}.
        \end{equation*}

        \item The projection $\Pi_{k,h}$ into test space $\mathcal{W}_{h}^{\alpha_2}$ must satisfy that 
        \begin{equation*}
            \ensuremath{\lim_{h\rightarrow0}\Vert\Pi_{k,h}\xi(\mathbf{k})-\xi(\mathbf{k})\Vert_{L^{2}(\Omega_{k}^{L})}=0},\ \forall \xi \in L^{2}(\Omega_{k}^{L}).
        \end{equation*}
    \end{enumerate}

     There is no need to specify particular projections until we implement them in the numerical examples, our method works with any of them.

\subsection{The Conservative Semi-Discrete Form}
Adopting the zero-flux boundary condition, testing the system on the cut-off domain with $\phi_{h}\in \mathcal{G}_{h}^{\alpha_1}$ and $\eta_{h} \in \mathcal{W}_{h}^{\alpha_2}$, we write the following semi-discrete weak form,
\begin{equation*}
    \begin{split}
        (\frac{\partial f_{h}}{\partial t},\phi_{h})_{p}&=-B^{u}_{L}(f_{h},W_{h},\phi_{h})\coloneqq-\iint_{\Omega^{L}_{k}\times \Omega^{L}_{p}} d\mathbf{k} d\mathbf{p} \{\mathcal{L}\phi_{h}\mathcal{L}f_{h}W_{h}\mathcal{B}\},\\
        (\frac{\partial W_{h}}{\partial t},\eta_{h})_{k}&=H^{u}_{L}(W_{h},f_{h},\eta_{h})\coloneqq\iint_{\Omega^{L}_{k}\times \Omega^{L}_{p}} d\mathbf{k} d\mathbf{p}\{\eta_{h}\mathcal{L}E\mathcal{L}f_{h}W_{h}\mathcal{B}\},
    \end{split}
\end{equation*}
where the subscript $L$ means integral on cut-off domain, the superscript $u$ means unconservative. We will first analyze the source of conservation errors and then present our conservative semi-discrete trilinear forms $B_{L}$ and $H_{L}$.

\subsubsection*{The Source of Conservation Errors}
Suppose different quadrature rules $R_{1}$ and $R_{2}$ are used for different equations,
\begin{equation*}
	\begin{split}
		(\frac{\partial f_{h}}{\partial t},\phi_{h})_{p}&=R_{1}\left[-\int_{\Omega_{p}^{L}}d\mathbf{p}\int_{\Omega_{k}^{L}}d\mathbf{k}\{\mathcal{L}\phi_{h}\mathcal{L}f_{h}W_{h}\mathcal{B}\}\right]\\
		(\frac{\partial W_{h}}{\partial t},\eta_{h})_{k}&=R_{2}\left[\int_{\Omega_{k}^{L}}d\mathbf{k}\int_{\Omega_{p}^{L}}d\mathbf{p}\{\eta_{h}\mathcal{L}E\mathcal{L}f_{h}W_{h}\mathcal{B}\}\right]
	\end{split}
\end{equation*}

The error of conservation laws can be decomposed into three terms,
\begin{equation*}
	\begin{split}
		&\frac{\partial}{\partial t}((f_{h},\Pi_{p,h}\phi)_{p}+(W_{h},\Pi_{k,h}\eta)_{k})\\
		=&R_{1}\left[-\int_{\Omega_{p}^{L}}d\mathbf{p}\int_{\Omega_{k}^{L}}d\mathbf{k}\left\{ \left(\mathcal{L}\Pi_{p,h}\phi\right)\mathcal{L}f_{h}W_{h}\mathcal{B}\right\} \right]+R_{2}\left[\int_{\Omega_{k}^{L}}d\mathbf{k}\int_{\Omega_{p}^{L}}d\mathbf{p}\left\{ \left(\Pi_{k,h}\eta\mathcal{L}E\right)\mathcal{L}f_{h}W_{h}\mathcal{B}\right\} \right]\\
		=&A_{1} + A_{2} + A_{3},
	\end{split}
\end{equation*}

where
\begin{equation*}
	\begin{split}
		A_{1}&=\left(R_{1}-I\right)\left[\int_{\Omega_{k}^{L}}d\mathbf{k}\int_{\Omega_{p}^{L}}d\mathbf{p}\left\{ \left(-\mathcal{L}\Pi_{p,h}\phi\right)\mathcal{L}f_{h}W_{h}\mathcal{B}\right\}\right],\\
		A_{2}&=\left(I-R_{2}\right)\left[\int_{\Omega_{k}^{L}}d\mathbf{k}\int_{\Omega_{p}^{L}}d\mathbf{p}\left\{ \left(-\mathcal{L}\Pi_{p,h}\phi\right)\mathcal{L}f_{h}W_{h}\mathcal{B}\right\}\right],\\
		A_{3}&=R_{2}\left[\int_{\Omega_{k}^{L}}d\mathbf{k}\int_{\Omega_{p}^{L}}d\mathbf{p}\left\{ \left(\Pi_{k,h}\eta\mathcal{L}E-\mathcal{L}\Pi_{p,h}\phi\right)\mathcal{L}f_{h}W_{h}\mathcal{B}\right\} \right].
	\end{split}
\end{equation*}

The error terms $A_{1}$ and $A_{2}$ are caused by inconsistent numerical integration on the resonance manifold. Suppose that $R_{1}-I$ is of the same order as $O(h^{a})$, and quadrature rule $R_{2}$ has error $O(h^{b})$, then the sum will be roughly $O(h^{\min\{a,b\}})$. The last error term $A_{3}$ is a result of projection error, whose order depends on the degree of test spaces, $\alpha_{1}$ and $\alpha_{2}$.

Note that $A_{1}$ and $A_{2}$ cancel out when we use the same quadrature rules, i.e. $R_{1}=R_{2}$. In what follows, we will introduce a conservative semi-discrete form such that $A_{3}$ disappears.

\subsubsection*{The Conservative Semi-Discrete Form}

Recall the definition of directional differential operator $\mathcal{L}$,
\begin{equation*}
	\mathcal{L}g=\frac{k_{\parallel}v_{\parallel}}{\omega}\frac{p}{p_{\parallel}}\frac{\partial g}{\partial p_{\parallel}}+(1-\frac{k_{\parallel}v_{\parallel}}{\omega})\frac{p}{p_{\perp}}\frac{\partial g}{\partial p_{\perp}}=N_{\parallel}\frac{\partial E}{\partial p_{\perp}}\frac{p}{p_{\perp}}\frac{\partial g}{\partial p_{\parallel}}+(1-N_{\parallel}\frac{\partial E}{\partial p_{\parallel}})\frac{p}{p_{\perp}}\frac{\partial g}{\partial p_{\perp}}.
\end{equation*}

We propose a discretized operator $\mathcal{L}_{h}$ defined as follows,

\begin{equation}
    \mathcal{L}_{h}g \coloneqq N_{\parallel,h}\frac{\partial E_{h}}{\partial p_{\perp}}\frac{p}{p_{\perp}}\frac{\partial g}{\partial p_{\parallel}}+(1-N_{\parallel,h}\frac{\partial E_{h}}{\partial p_{\parallel}})\frac{p}{p_{\perp}}\frac{\partial g}{\partial p_{\perp}},
\label{discL}
\end{equation}

where the discretized kinetic energy is defined as $E_{h}=\Pi_{p,h}E(\mathbf{p})$, and the discretized wave refractive index is defined as $N_{\parallel,h} = \Pi_{k,h} N_{\parallel} = \Pi_{k,h} \frac{k_{\parallel}}{\omega(\mathbf{k})}$.

The main result of this paper is stated in the following theorem.
\begin{theorem} \label{cons_thm}
\label{semid_conserv}
If $f_h(\mathbf{p},t)$ and $W_h(\mathbf{k},t)$ are solutions of the following semi-discrete weak form, 

\begin{equation}
    \begin{split}
        (\frac{\partial f_{h}}{\partial t},\phi_{h})_{p}&=-B_{L}(f_{h},W_{h},\phi_{h})\coloneqq-\iint_{\Omega^{L}_{k}\times \Omega^{L}_{p}} d\mathbf{k} d\mathbf{p} \{\mathcal{L}_{h}\phi_{h}\mathcal{L}_{h}f_{h}W_{h}\mathcal{B}_h\},\\
        (\frac{\partial W_{h}}{\partial t},\eta_{h})_{k}&=H_{L}(W_{h},f_{h},\eta_{h})\coloneqq\iint_{\Omega^{L}_{k}\times \Omega^{L}_{p}} d\mathbf{k} d\mathbf{p}\{\eta_{h}\mathcal{L}_{h}E_{h}\mathcal{L}_{h}f_{h}W_{h}\mathcal{B}_h\},
    \end{split}
    \label{semidisc}
\end{equation}

then the following discrete conservation laws hold,
\begin{equation*}
    \begin{split}
        \frac{\partial}{\partial t}\mathcal{M}_{tot,h}&=\frac{\partial}{\partial t}((f_{h},\Pi_{p,h} 1)_{p}+(W_{h},0)_{k})=0,\\
        \frac{\partial}{\partial t}\mathcal{P}^{\parallel}_{tot,h}&=\frac{\partial}{\partial t}\left((f_{h},\Pi_{p,h} p_{\parallel})_{p}+(W_{h}, \Pi_{k,h} N_{\parallel})_{k}\right)=0,\\
        \frac{\partial}{\partial t}\mathcal{E}_{tot,h}&=\frac{\partial}{\partial t}((f_{h}, \Pi_{p,h} E(\mathbf{p}))_{p}+(W_{h}, \Pi_{k,h} 1)_{k})=0.
    \end{split}
\end{equation*}
\end{theorem}

\begin{proof}
Substitute the discrete conservation pairs $\{\Pi_{p,h} 1, 0\}$, $\{\Pi_{p,h} p_{\parallel}, \Pi_{k,h} N_{\parallel}\}$ and $\{\Pi_{p,h} E(\mathbf{p}),  \Pi_{k,h} 1\}$ into semi-discrete form (\ref{semidisc}) and use the definition of $\mathcal{L}_{h}$.
\end{proof}

\begin{coro} \label{orthog}
If in addition to the assumptions of Theorem(\ref{cons_thm}), the projections are $L^2$ orthogonal projections, i.e.
\begin{equation*}
	\left(u-\Pi_{p,h}u,v\right)_{p}=0,\forall v\in\mathcal{G}_{h}^{\alpha_{1}}
\end{equation*}
and 
\begin{equation*}
	\left(U-\Pi_{k,h}U,V\right)_{k}=0,\forall V\in\mathcal{W}_{h}^{\alpha_{2}},
\end{equation*}
then the exact conservation laws are preserved, i.e.
\begin{equation*}
    \begin{split}
        \frac{\partial}{\partial t}\mathcal{M}_{tot,h}&=\frac{\partial}{\partial t}((f_{h},1)_{p}+(W_{h},0)_{k})=0,\\
        \frac{\partial}{\partial t}\mathcal{P}^{\parallel}_{tot,h}&=\frac{\partial}{\partial t}\left((f_{h},p_{\parallel})_{p}+(W_{h}, N_{\parallel})_{k}\right)=0,\\
        \frac{\partial}{\partial t}\mathcal{E}_{tot,h}&=\frac{\partial}{\partial t}((f_{h},E(\mathbf{p}))_{p}+(W_{h},1)_{k})=0.
    \end{split}
\end{equation*}
\end{coro}

\begin{proof}
Use the fact that $f_{h} \in\mathcal{G}_{h}^{\alpha_{1}}$ and $W_{h} \in\mathcal{W}_{h}^{\alpha_{2}}$.
\end{proof}
\begin{remark}
Same as stated in Theorem(\ref{uncond}), our semi-discrete weak form is also unconditionally conservative, i.e. the conservation does not depend on a particular discrete emission/absorption kernel $\mathcal{B}_{h}$.
\end{remark}
\bigskip

\section{The Sparse Interaction Tensors} \label{tensor}
Suppose that the test spaces are spanned by basis functions, i.e. $\mathcal{G}_{h}^{\alpha_{1}}=\text{span}\{\phi_{i}\}$
and $\mathcal{W}_{h}^{\alpha_{2}}=\text{span}\{\eta_{j}\}$. Then we can express the discrete particle \textit{pdf} $f_{h}$ and wave \textit{sed} $W_{h}$ as a linear combination of basis functions.

\begin{equation*}
    \begin{split}
        f_{h}(\mathbf{p},t)&=\sum_{i=1}^{N_{f}}a_{i}(t)\phi_{i}(\mathbf{p})\\
        W_{h}(\mathbf{k},t)&=\sum_{j=1}^{N_{w}}w_{j}(t)\eta_{j}(\mathbf{k})
    \end{split}
\end{equation*}

By definition, $E_{h} = \Pi_{p,h} E \in\mathcal{G}_{h}^{\alpha_{1}}$, therefore it is also a linear
combination of basis functions, $E_{h}=\sum_{q=1}^{N_{f}}E_{q}\phi_{q}$.

Substitute the above expressions into Equation(\ref{semidisc}), then the semi-discrete system becomes a first-order finite dimension ODE system:

\begin{equation*}
    \begin{split}
        \sum_{i=1}^{N_{f}}\frac{\partial a_{i}}{\partial t}\int\phi_{i}\phi_{m}d^{3}p&=-\sum_{n=1}^{N_{f}}\sum_{k=1}^{N_{w}}a_{n}w_{k}\int_{\Omega^{L}_{p}}d\mathbf{p}\int_{\Omega^{L}_{k}}d\mathbf{k}\{\mathcal{L}_{h}\phi_{m}\mathcal{L}_{h}\phi_{n}\eta_{k}\mathcal{B}(\mathbf{p};\mathbf{k})\}\\
        \sum_{j=1}^{N_{w}}\frac{\partial w_{j}}{\partial t}\int\eta_{j}\eta_{q}d^{3}k&=\sum_{k=1}^{N_{w}}\sum_{n=1}^{N_{f}}w_{k}a_{n}\int_{\Omega^{L}_{k}}d\mathbf{k}\int_{\Omega^{L}_{p}}d\mathbf{p}\{\eta_{q}\mathcal{L}_{h}E_{h}\mathcal{L}_{h}\phi_{n}\eta_{k}\mathcal{B}(\mathbf{p};\mathbf{k})\}
    \end{split}
\end{equation*}

\bigskip

Denote the mass matrix for particle \textit{pdf} as $A_{im}=(\phi_{i},\phi_{m})_{p}$, and denote the mass matrix for wave \textit{sed} as $G_{jq}=(\eta_{j},\eta_{q})_{k}$. 

Analogously, define the interaction tensors $B$ and $H$ corresponding to the trilinear forms. 

\begin{equation*}
    \begin{split}
        B_{nkm}&=B(\phi_{n},\eta_{k},\phi_{m})\\
        H_{knq}&=H(\eta_{k},\phi_{n},\eta_{q})
    \end{split}
\end{equation*}

As a result, we obtain the nonlinear ODE system corresponding to semi-discrete weak form(\ref{semidisc}):

\begin{equation}
    \begin{split}
        \frac{\partial a_{i}}{\partial t}A_{im}&=-a_{n}w_{k}B_{nkm}\\
        \frac{\partial w_{j}}{\partial t}G_{jq}&=w_{k}a_{n}H_{knq}
    \end{split}
    \label{ODEsys}
\end{equation}

The interaction tensors $B$ and $H$ are both sparse tensors for two reasons: compactly supported basis and the resonant feature of trilinear forms. Taking particle interaction tensor $B$ as an example, $B_{nkm}=0$ when
\begin{enumerate}
    \item $\phi_{m}$ and $\phi_{n}$ are not in neighboring elements.
    \item $\phi_{m}$ and $\eta_{k}$ do not "resonate", i.e. $\text{supp} (\phi_{m}) \times \text{supp} (\eta_{k})$ does not intersect with the resonant manifold.  
\end{enumerate}

Suppose in each dimension we have $O(n)$ meshes, then the shape of particle interaction tensor $B$ is roughly $O(n^{2}) \times O(n^{2}) \times O(n^{2})$, while the number of nonzero elements will be only $O(n^3)$, i.e. the sparsity of tensor $B$ is about $1-\frac{1}{O(n^{3})}$. A similar analysis can also be applied to the wave interaction tensor $H$. 

\bigskip
We observed that the trilinear forms $B$ and $H$ defined in Equation(\ref{trilin}) have similar structures. Therefore one might wonder if there is any relation between the interaction tensors $B$ and $H$. It turns out that when $\alpha_{2}=0$, i.e piecewise constant basis functions are used for wave \textit{sed} $W_h$, we can infer any nonzero element of wave interaction tensor $H$ from particle interaction tensor $B$. In practice, the interaction tensors are precomputed and saved for later use. Taking advantage of this relation, we can save half the time of preprocessing. The derivation is as follows.

When $\alpha_{2}=0$, $\mathcal{W}_{h}^{\alpha_{2}}=\text{span}\{\eta_{j}\}$
are piecewise constant functions, we have
\begin{equation*}
    \eta_{i}(\mathbf{k})\eta_{j}(\mathbf{k})=\delta_{ij}\eta_{i}(\mathbf{k}).
\end{equation*}
Then the mass matrix for wave \textit{sed} is diagonal,
\begin{equation*}
    G_{jq}=(\eta_{j},\eta_{q})_{k}=\int_{R_{k}^{j}}\delta_{jq}d^{3}k=\text{diag}(\mu(R_{k}^{j})),
\end{equation*}
where $\mu(R_{k}^{j})=\int_{R_{k}^{j}}1d^{3}k$ is the measure of
$j$-th element in $\Omega^{L}_{k}$.

Moreover, note that if we define a $4$-th order tensor
\begin{equation*}
    \tilde{H}_{mknq}\coloneqq\int_{\Omega^{L}_{k}}d\mathbf{k}\int_{\Omega^{L}_{p}}d\mathbf{p}\{\mathcal{L}_{h}\phi_{m}\mathcal{L}_{h}\phi_{n}\eta_{k}\eta_{q}\mathcal{B}(\mathbf{p};\mathbf{k})\}
\end{equation*}

Recall the expansion $E_{h}=\sum_{q=1}^{N_{f}}E_{q}\phi_{q}$, and substitute it into the definition of wave interaction tensor $H$, we obtain the relation between $H_{knq}$ and $\tilde{H}_{mknq}$,
\begin{equation*}
    H_{knq}=H(\eta_{k},\phi_{n},\eta_{q})=\sum_{m=1}^{N_{f}}E_{m}\tilde{H}_{mknq}
\end{equation*}

It can be observed that the form of $\tilde{H}_{mknq}$ is almost identical to the definition of particle tensor $B_{mnk}$, except for the extra $\eta_{q}$. Replace $\eta_{k}(\mathbf{k})\eta_{q}(\mathbf{k})$ with $\delta_{kq}\eta_{k}(\mathbf{k})$, we obtain the relation between $\tilde{H}_{mknq}$ and $B_{mnk}$,
\begin{equation*}
    \tilde{H}_{mknq}=\delta_{kq}\int_{\Omega^{L}_{k}}d\mathbf{k}\int_{\Omega^{L}_{p}}d\mathbf{p}\{\mathcal{L}_{h}\phi_{m}\mathcal{L}_{h}\phi_{n}\eta_{k}\mathcal{B}(\mathbf{p};\mathbf{k})\}=\delta_{kq}B_{mnk}.
\end{equation*}
Therefore $B_{mnk}$ and $E_{m}$ is all we need to calculate $H_{knq}$,
\begin{equation*}
    H_{knq}=\sum_{m=1}^{N_{f}}E_{m}\tilde{H}_{mknq}=\sum_{m=1}^{N_{f}}E_{m}\delta_{kq}B_{mnk}=\delta_{kq}\sum_{m=1}^{N_{f}}E_{m}B_{mnk}=\begin{cases}
    0, & k\neq q\\
    \sum_{m=1}^{N_{f}}E_{m}B_{mnk} & k=q
    \end{cases}
\end{equation*}
\bigskip

\section{Stability and Positivity} \label{stability}
In this section, we investigate the stability of the fully discretized nonlinear system. With semi-implicit time discretization, there is no constraint on time step size from the CFL condition. However, the stability will rely on the positivity of $W_{h}$, which results in a condition for the time step size, relevant to the gradient of particle \textit{pdf} $f_{h}$. The condition will not cause any trouble for implementation, because we can always adapt the step size a posteriori.

\subsection{Stability of the Semi-Discrete Form}
Consider the equation for particle \textit{pdf} only, it has the form of a diffusion equation, thus its stability relies on the fact that the diffusion coefficient is positive semi-definite, which further relies on the positivity of wave \textit{sed} $W_{h}$.

\begin{lemma}[$L^2$ stability of $f_h(\mathbf{p})$ and $L^1$ bound of $W_h(\mathbf{k})$]

Suppose $f_{h}(\mathbf{p},t)$ and $W_{h}(\mathbf{k},t)$ are the solution
of equation(\ref{semidisc}) with the following initial condition:
\begin{equation*}
    \begin{split}
        f_{h}(\mathbf{p},0)&=f_{h}^{0}(\mathbf{p})\\
        W_{h}(\mathbf{k},0)&=W_{h}^{0}(\mathbf{k})
    \end{split}
\end{equation*}

If $W_{h}$ always takes non-negative values, i.e. $W_{h}(\mathbf{k},t)\geq0,\forall~\mathbf{k}\in\Omega^{L}_{k}, \forall t\geq0$,
then $f_{h}$ has $L^{2}$ stability 
\begin{equation*}
    \Vert f_{h}\Vert _{L^{2}(\Omega^{L}_{p})}\leq\Vert f_{h}^{0}\Vert _{L^{2}(\Omega^{L}_{p})}
\end{equation*}

and $W_{h}$ has bounded $L^{1}$ norm.
\begin{equation*}
    \Vert W_{h}\Vert _{L^{1}(\Omega^{L}_{k})}\leq\mathcal{E}_{tot,h}^{0}+\Vert f_{h}^{0}\Vert _{L^{2}(\Omega^{L}_{p})}\cdot\Vert \Pi_{p,h} E\Vert _{L^{2}(\Omega^{L}_{p})}
\end{equation*}

\label{lemma_stb}
\end{lemma}

\begin{proof}
Since $f_{h}$ belongs to the test space $\mathcal{G}_{h}^{\alpha_{1}}$, we test the equation for particles with $f_{h}$, we obtain that
\begin{equation*}
    (\frac{\partial f_{h}}{\partial t},f_{h})_{p}=-\int_{\Omega^{L}_{k}}d\mathbf{k}\int_{\Omega^{L}_{p}}d\mathbf{p}\{(\mathcal{L}_{h}f_{h})^{2}W_{h}\mathcal{B}\}.
\end{equation*}

The right hand side is non-positive as long as $W_{h}$ always take non-negative values, therefore the $L^2$ norm of $f_{h}$ always decreases,
\begin{equation*}
    \frac{1}{2}\frac{\partial}{\partial t}\Vert f_{h}\Vert _{L^{2}(\Omega^{L}_{p})}^{2}\leq0\Rightarrow\Vert f_{h}\Vert _{L^{2}(\Omega^{L}_{p})}\leq\Vert f_{h}^{0}\Vert _{L^{2}(\Omega^{L}_{p})}
\end{equation*}

\bigskip

Now consider $W_{h}$, by definition,
\begin{equation*}
    \Vert W_{h}\Vert _{L^{1}(\Omega^{L}_{k})}=(W_{h},\text{sgn}(W_{h}))_{k}=(W_{h},1)_{k}
\end{equation*}

Recall the energy conservation property in Theorem \ref{cons_thm}:
\begin{equation*}
    \Vert W_{h}\Vert _{L^{1}(\Omega^{L}_{k})}+(f_{h},\Pi_{p,h} E)_{p}=\mathcal{E}_{tot,h}^{0}
\end{equation*}

Use Holder's inequality
\begin{equation*}
    \begin{split}
        \Vert W_{h}\Vert_{L^{1}(\Omega^{L}_{k})}&=\mathcal{E}_{tot,h}^{0}-(f_{h},\Pi_{p,h}E)_{p}\\
        &\leq\mathcal{E}_{tot,h}^{0}+|(f_{h},\Pi_{p,h} E)_{p}|\\
        &\leq\mathcal{E}_{tot,h}^{0}+\Vert f_{h}\Vert _{L^{2}(\Omega^{L}_{p})}\cdot\Vert \Pi_{p,h} E\Vert _{L^{2}(\Omega^{L}_{p})}
    \end{split}
\end{equation*}

By the $L^{2}$-stability of $f_{h}$, we obtain the upper bound of $W_{h}$'s $L^1$ norm,
\begin{equation*}
    \Vert W_{h}\Vert _{L^{1}(\Omega^{L}_{k})}\leq\mathcal{E}_{tot,h}^{0}+\Vert f_{h}\Vert _{L^{2}(\Omega^{L}_{p})}\cdot\Vert E\Vert _{L^{2}(\Omega^{L}_{p})}\leq\mathcal{E}_{tot,h}^{0}+\Vert f_{h}^{0}\Vert _{L^{2}(\Omega^{L}_{p})}\cdot\Vert \Pi_{p,h} E\Vert _{L^{2}(\Omega^{L}_{p})}
\end{equation*}

\end{proof}

\subsection{Time Discretization}
Recall our conservative semi-discrete weak form,
\begin{equation*}
    \begin{split}
        (\frac{\partial f_{h}}{\partial t},\phi_{h})_{p}&=-B_{L}(f_{h},W_{h},\phi_{h})\coloneqq-\iint_{\Omega^{L}_{k}\times \Omega^{L}_{p}} d\mathbf{k} d\mathbf{p} \{\mathcal{L}_{h}\phi_{h}\mathcal{L}_{h}f_{h}W_{h}\mathcal{B}_h\},\\
        (\frac{\partial W_{h}}{\partial t},\eta_{h})_{k}&=H_{L}(W_{h},f_{h},\eta_{h})\coloneqq\iint_{\Omega^{L}_{k}\times \Omega^{L}_{p}} d\mathbf{k} d\mathbf{p}\{\eta_{h}\mathcal{L}_{h}E_{h}\mathcal{L}_{h}f_{h}W_{h}\mathcal{B}_h\},
    \end{split}
\end{equation*}

The time step size of the explicit scheme for diffusion equations is restricted by the CFL condition. Two reasons urge us to avoid explicit schemes,
\begin{enumerate}
    \item The CFL bound of step size may be too restrictive, and we might lose efficiency.
    \item The upper bound depends on the eigenvalues of time-varying diffusion coefficients. However, in the proposed scheme, we never calculate the diffusion coefficient explicitly, instead, we compute the interaction tensor $B$ associated with the trilinear form $B$. 
\end{enumerate}

On the other hand, due to nonlinearity, a fully implicit scheme requires fixed-point iteration involving both particle \textit{pdf} $f_{h}$ and wave \textit{sed} $W_{h}$, which can be time-consuming. Therefore, the objective is to find a scheme that is only implicit for $f_{h}$, and at the same time preserves discrete conservation laws. 

We propose the following semi-implicit scheme,
\begin{equation}
    \begin{split}
        (\frac{f_{h}^{s+1}-f_{h}^{s}}{\Delta t},\phi_{h})_{p} & +B_{L}(f_{h}^{s+1},W_{h}^{s},\phi_{h})=0\\
        (\frac{W_{h}^{s+1}-W_{h}^{s}}{\Delta t},\eta_{h})_{k} & -H_{L}(W_{h}^{s},f_{h}^{s+1},\eta_{h})=0
    \end{split}
    \label{semiimp}
\end{equation}

The scheme is implicit for particle \textit{pdf} $f_{h}$ if we focus on the first line, meanwhile it is explicit for wave \textit{sed} $W_{h}$, considering the second line. For implementation, we solve the first row and then substitute the next step particle \textit{pdf} $f^{s+1}_{h}$ into the second row. It can be easily verified that the discrete conservation laws still hold, i.e. we have
\begin{equation*}
    (f_{h}^{s+1},\phi_{c,h})_{p}+(W_{h}^{s+1},\eta_{c,h})_{k}=(f_{h}^{s},\phi_{c,h})_{p}+(W_{h}^{s},\eta_{c,h})_{k}
\end{equation*}

The following theorem is the fully discrete version of Lemma \ref{lemma_stb}, giving the unconditional $L^2$-stability of $f^{s}_{h}$ when $W^{s}_{h}$ is non-negative.

\begin{theorem}
Suppose $f^{s}_{h}(\mathbf{p})$ and $W^{s}_{h}(\mathbf{k})$ are the solution
of Equation(\ref{semiimp}).

If $W^{s}_{h}$ always takes non-negative values, i.e. $W^{s}_{h}(\mathbf{k})\geq0,\forall~\mathbf{k}\in\Omega^{L}_{k}, \forall s\geq0$,
then $f^{s}_{h}$ has $L^{2}$ stability 
\begin{equation*}
    \Vert f^{s}_{h}\Vert _{L^{2}(\Omega^{L}_{p})}\leq\Vert f_{h}^{0}\Vert _{L^{2}(\Omega^{L}_{p})}
\end{equation*}

and $W^{s}_{h}$ has bounded $L^{1}$ norm.
\begin{equation*}
    \Vert W^{s}_{h}\Vert _{L^{1}(\Omega^{L}_{k})}\leq\mathcal{E}_{tot,h}^{0}+\Vert f_{h}^{0}\Vert _{L^{2}(\Omega^{L}_{p})}\cdot\Vert \Pi_{p,h} E\Vert _{L^{2}(\Omega^{L}_{p})}
\end{equation*}
\label{thm_stb}
\end{theorem}

\begin{proof}
    Given that $W_{h}^{s}(\mathbf{k})\geq0, \forall \mathbf{k}\in \Omega^{L}_{k}$, we have $B_{L}(f_{h}^{s+1},W_{h}^{s},f_{h}^{s+1})\geq0$. Therefore, $f^{s}_{h}$ has unconditional $L^2$-stability.

    Since the scheme(\ref{semiimp}) preserves energy conservation, the $L^1$ bound of $W_{h}$ can be proved in the same approach as we have done in Lemma \ref{lemma_stb}.
\end{proof}

Note that the stability depends on our assumption that $W^{s}_{h}$ is non-negative. Therefore, in what follows, we will discuss the positivity-preserving technique of $W^{s}_{h}$.

\subsection{Positivity-Preserving Technique for the Wave SED}
To ensure positivity of wave \textit{sed} $W^{s}_{h}$, we draw the strategy from Zhang et al.\cite{zhang2010maximum}:
\begin{enumerate}
    \item Use a small enough time step to ensure positive cell-average of a temporary wave \textit{sed} $W^{s+1,*}_{h}$, given that we have pointwise positivity of last step wave \textit{sed} $W^{s}_{h}$.
    \item Apply a slope limiter on $W^{s+1,*}_{h}$ which preserves cell-average at the same time, then we obtain a pointwise positive $W^{s+1}_{h}$ as our solution of the next step wave \textit{sed}. (Obviously, if we use piecewise constant basis functions, this step is not necessary).
\end{enumerate}

Firstly we will derive the constraint on time step size. After that, we explain why the slope limiter will not break discrete conservation laws.

\bigskip
Suppose $\eta_{j,0}$ is the characteristic function of the $j$-th
element $R_{k}^{j}\subset\Omega^{L}_{k}$, i.e. $\eta_{j,0}=1_{\mathbf{k}\in R_{k}^{j}}$, which belongs to the test space $\mathcal{W}_{h}^{\alpha_2}$. According to the time discretization in Equation(\ref{semiimp}),

\begin{equation*}
    \begin{split}
        (\frac{W_{h}^{s+1,*}-W_{h}^{s}}{\Delta t},\eta_{j,0})_{k}=\int_{\Omega^{L}_{k}}\frac{W_{h}^{s+1,*}-W_{h}^{s}}{\Delta t}\eta_{j,0}d\mathbf{k}=\int_{\Omega^{L}_{k}}d\mathbf{k}\int_{\Omega^{L}_{p}}d\mathbf{p}\{\mathcal{L}_{h}E_{h}\mathcal{L}_{h}f_{h}^{s+1}W_{h}^{s}\eta_{j,0}\mathcal{B}_{h}\}
    \end{split},
\end{equation*}
which is equivalent to 
\begin{equation*}
    \int_{R_{k}^{j}}W_{h}^{s+1,*}d^{3}k=\int_{R_{k}^{j}}W_{h}^{s}(1+\Delta t\int_{\Omega^{L}_{p}}d\mathbf{p}\{\mathcal{L}_{h}E_{h}\mathcal{L}_{h}f_{h}^{s+1}\mathcal{B}_{h}\})d\mathbf{k}
\end{equation*}

To ensure positive cell-average, i.e. $\int_{R_{k}^{j}}W_{h}^{s+1,*}d\mathbf{k} \geq 0$, we require that there exists a constant $\epsilon>0$ such that
\begin{equation}
    1+\Delta t\int_{\Omega^{L}_{p}}d\mathbf{p}\{\mathcal{L}_{h}E_{h}\mathcal{L}_{h}f_{h}^{s+1}\mathcal{B}_{h}\} \geq \epsilon,\ \forall \mathbf{k}\in R_{k}^{j}
    \label{stepsize}
\end{equation}

As long as the time step size $\Delta t$ satisfy condition(\ref{stepsize}), we have $\int_{R_{k}^{j}}W_{h}^{s+1,*}d\mathbf{k}\geq\epsilon\int_{R_{k}^{j}}W_{h}^{s}d\mathbf{k}\geq0$.

The following theorem guarantees that our bound for $\Delta t$ will not shrink over time.
\begin{theorem}
    For any $\epsilon>0$, given a regular enough discrete emission/absorption kernel $\mathcal{B}_{h}$, there exists a constant $\Delta t_{M}$ determined by $\epsilon$, $f^{0}$, $\Omega^{L}_{p}$, $\Omega^{L}_{k}$ and $h$, such that any $\Delta t < \Delta t_{M}$ satisfies condition (\ref{stepsize}).
\end{theorem}

\begin{proof}
By H\"older's inequality, the "growth rate" is bounded as follows,
\begin{equation}
    \begin{split}
        \left\vert\int_{\Omega_{p}^{L}}d\mathbf{p}\{\mathcal{L}_{h}E_{h}\mathcal{L}_{h}f_{h}^{s+1}\mathcal{B}_{h}\}\right\vert\leq&\int_{\Omega_{p}^{L}}d\mathbf{p}\left\vert\mathcal{L}_{h}E_{h}\mathcal{L}_{h}f_{h}^{s+1}\mathcal{B}_{h}\right\vert\\
        \leq&\left\Vert\frac{p_{\perp}}{p}\mathcal{L}_{h}f_{h}^{s+1}\right\Vert_{L^{\infty}(\Omega_{p}^{L})}\cdot\left\Vert\frac{p}{p_{\perp}}\mathcal{B}_{h}\mathcal{L}_{h}E_{h}\right\Vert_{L^{1}(\Omega_{p}^{L})}\\
        =&2\pi\left\Vert \frac{p_{\perp}}{p}\mathcal{L}_{h}f_{h}^{s+1}\right\Vert_{L^{\infty}(\Omega_{p}^{L})}\cdot\left(\int p \mathcal{B}_{h}\mathcal{L}_{h}E_{h}dp_{\perp}dp_{\parallel}\right)
    \end{split}
    \label{holder_1_inf}
\end{equation}

Firstly, consider the $L^{\infty}$-norm factor from inequality (\ref{holder_1_inf}). Recall the definition of $\mathcal{L}_{h}$,
\begin{equation*}
    \frac{p_{\perp}}{p}\mathcal{L}_{h}f_{h}^{s+1}\coloneqq N_{\parallel,h}\frac{\partial E_{h}}{\partial p_{\perp}}\frac{\partial f_{h}^{s+1}}{\partial p_{\parallel}}+(1-N_{\parallel,h}\frac{\partial E_{h}}{\partial p_{\parallel}})\frac{\partial f_{h}^{s+1}}{\partial p_{\perp}}.
\end{equation*}

Both of the coefficients $N_{\parallel,h}\frac{\partial E_{h}}{\partial p_{\perp}}$
and $(1-N_{\parallel,h}\frac{\partial E_{h}}{\partial p_{\parallel}})$
are bounded by some constant $C$ dependent on $\Omega_{p}^{L}$ and $\Omega_{k}^{L}$, hence it follows that,

\begin{equation}
    \begin{split}
        \left\Vert\frac{p_{\perp}}{p}\mathcal{L}_{h}f_{h}^{s+1}\right\Vert_{L^{\infty}(\Omega_{p}^{L})}
        &=\left\Vert N_{\parallel,h}\frac{\partial E_{h}}{\partial p_{\perp}}\frac{\partial f_{h}^{s+1}}{\partial p_{\parallel}}+(1-N_{\parallel,h}\frac{\partial E_{h}}{\partial p_{\parallel}})\frac{\partial f_{h}^{s+1}}{\partial p_{\perp}}\right\Vert_{L^{\infty}(\Omega_{p}^{L})}\\
        &\leq C(\Omega_{p}^{L},\Omega_{k}^{L})\left\Vert\frac{\partial f_{h}^{s+1}}{\partial p_{\parallel}}\right\Vert_{L^{\infty}(\Omega_{p}^{L})}+C(\Omega_{p}^{L},\Omega_{k}^{L})\left\Vert\frac{\partial f_{h}^{s+1}}{\partial p_{\perp}}\right\Vert_{L^{\infty}(\Omega_{p}^{L})}\\
        &\leq C_{1}(\Omega_{p}^{L},\Omega_{k}^{L})\cdot\left\Vert\nabla_{p}f_{h}^{s+1}\right\Vert_{L^{\infty}(\Omega_{p}^{L})},
    \end{split}
    \label{Linfbound}
\end{equation}

We claim that $\left\Vert\nabla_{p}f_{h}^{s+1}\right\Vert_{L^{\infty}(\Omega_{p}^{L})}$ is bounded uniformly in time. Indeed, since the domain $\Omega_{p}^{L}$ is finite, all $L^{r}$ norms are equivalent, therefore,

\begin{equation*}
    \left\Vert\nabla_{p}f_{h}^{s+1}\right\Vert_{L^{\infty}(\Omega_{p}^{L})}\leq C_{2}( \Omega_{p}^{L})\left\Vert\nabla_{p}f_{h}^{s+1}\right\Vert_{L^{2}(\Omega_{p}^{L})}.
\end{equation*}

Moreover, the inverse inequality for finite element spaces,
\begin{equation*}
    \left\Vert\nabla_{p}f_{h}^{s+1}\right\Vert_{L^{2}(\Omega_{p}^{L})}\leq\frac{C_{3}}{h_{p}}\left\Vert f_{h}^{s+1}\right\Vert_{L^{2}(\Omega_{p}^{L})}, 
\end{equation*}
and the $L^{2}$ stability estimate from Theorem \ref{thm_stb}
\begin{equation*}
    \left\Vert f_{h}^{s+1}\right\Vert_{L^{2}(\Omega_{p}^{L})}\leq\left\Vert f_{h}^{0}\right\Vert_{L^{2}(\Omega_{p}^{L})},
\end{equation*}
 leads to the following estimate for $\nabla_{p}f_{h}^{s+1}$,
\begin{equation}
    \left\Vert\nabla_{p}f_{h}^{s+1}\right\Vert_{L^{\infty}(\Omega_{p}^{L})}
        \leq C_{2}(\Omega_{p}^{L})\frac{C_{3}}{h_{p}}\left\Vert f_{h}^{0}\right\Vert_{L^{2}(\Omega_{p}^{L})}.
    \label{discretebound}
\end{equation}

Therefore,  the $L^{\infty}$-norm factor from inequality (\ref{holder_1_inf}) is bounded as follows,

\begin{equation*}
    \left\Vert\frac{p_{\perp}}{p}\mathcal{L}_{h}f_{h}^{s+1}\right\Vert_{L^{\infty}(\Omega_{p}^{L})} \leq C_{1}(\Omega_{p}^{L},\Omega_{k}^{L})C_{2}(\Omega_{p}^{L})\frac{C_{3}}{h_{p}}\left\Vert f_{h}^{0}\right\Vert_{L^{2}(\Omega_{p}^{L})}.
\end{equation*}

\bigskip
Next, consider the $L^{1}$-norm factor from inequality (\ref{holder_1_inf}), and write it as follows,

\begin{equation*}
    \int p\mathcal{B}_{h}\mathcal{L}_{h}E_{h}dp_{\perp}dp_{\parallel}=\int\left[pU_{l}(\mathbf{p};\mathbf{k})\delta_{h}(\omega(\mathbf{k})-k_{\parallel}v_{\parallel}-l\omega_{c}/\gamma(\mathbf{p}))\frac{p}{p_{\perp}}\frac{\partial E_{h}}{\partial p_{\perp}}\right]dp_{\perp}dp_{\parallel},
\end{equation*}

where $\delta_{h}$ represents an approximation of Dirac delta, see Equation (\ref{ATI}) and $\ref{MSA}$.

We discuss the following two cases,
\begin{itemize}
    \item When $l\neq 0$, since 
    \begin{equation*}
        \lim_{p_{\perp}\rightarrow0^{+}}U_{l}(\mathbf{p};\mathbf{k})\frac{p^{2}}{p_{\perp}}=\lim_{p_{\perp}\rightarrow0^{+}}8\pi^{2}e^{2}\frac{\left(iE_{2}J_{l}^{\prime}\right)^{2}\left(\frac{p_{\perp}}{p}\right)^{2}}{\left(1-E_{2}^{2}\right)\frac{1}{\omega}\frac{\partial}{\partial\omega}(\omega^{2}\varepsilon)+2iE_{2}\frac{1}{\omega}\frac{\partial}{\partial\omega}(\omega^{2}g)+E_{3}^{2}\frac{1}{\omega}\frac{\partial}{\partial\omega}(\omega^{2}\eta)}\frac{p^{2}}{p_{\perp}}=0,
    \end{equation*}
    the integral is bounded as follows,
    \begin{equation}
        \begin{split}
            \int p\mathcal{B}_{h}\mathcal{L}_{h}E_{h}dp_{\perp}dp_{\parallel}
            &=\int\left[\left(U_{l}(\mathbf{p};\mathbf{k})\frac{p^{2}}{p_{\perp}}\right)\frac{\partial E_{h}}{\partial p_{\perp}}\delta_{h}(\omega(\mathbf{k})-k_{\parallel}v_{\parallel}-l\omega_{c}/\gamma(\mathbf{p}))\right]dp_{\perp}dp_{\parallel}\\
            &\leq\sup_{\Omega_{p}^{L}}\left(U_{l}(\mathbf{p};\mathbf{k})\frac{p^{2}}{p_{\perp}}\right)\frac{\partial E_{h}}{\partial p_{\perp}}\int\left[\delta_{h}(\omega(\mathbf{k})-k_{\parallel}v_{\parallel}-l\omega_{c}/\gamma(\mathbf{p}))\right]dp_{\perp}dp_{\parallel}\\
            &\leq C_{4}(\Omega_{p}^{L},\Omega_{k}^{L}).
        \end{split}
        \label{L1bound}
    \end{equation}
    \item When $l = 0$, the above trick does not work, because $\lim_{p_{\perp}\rightarrow0^{+}}U_{0}(\mathbf{p};\mathbf{k})>0$. For this special case, as an alternative to the original operator 
    \begin{equation*}
        \mathcal{L}_{h}g \coloneqq N_{\parallel,h}\frac{\partial E_{h}}{\partial p_{\perp}}\frac{p}{p_{\perp}}\frac{\partial g}{\partial p_{\parallel}}+(1-N_{\parallel,h}\frac{\partial E_{h}}{\partial p_{\parallel}})\frac{p}{p_{\perp}}\frac{\partial g}{\partial p_{\perp}},
    \end{equation*}
    we adopt a new discrete operator,
    \begin{equation*}
        \mathcal{L}^{0}_{h}g\coloneqq N_{\parallel,h}\frac{\partial E_{h}}{\partial p_{\perp}}\frac{p}{E_{h}\frac{\partial E_{h}}{\partial p_{\perp}}}\frac{\partial g}{\partial p_{\parallel}}+(1-N_{\parallel,h}\frac{\partial E_{h}}{\partial p_{\parallel}})\frac{p}{E_{h}\frac{\partial E_{h}}{\partial p_{\perp}}}\frac{\partial g}{\partial p_{\perp}}.
    \end{equation*}

    The operator is still a consistent discretization since $E=\sqrt{1+p_{\parallel}^{2}+p_{\perp}^{2}}$. 
    It can also be easily verified that the $L^{\infty}$ bound in inequality (\ref{Linfbound}) is still true with this new operator $\mathcal{L}^{0}_{h}$. 

    In addition, the $L^{1}$-norm factor becomes,
    \begin{equation*}
        \int p\mathcal{B}_{h}\mathcal{L}_{h}^{0}E_{h}dp_{\perp}dp_{\parallel}=\int\left[pU_{l}(\mathbf{p};\mathbf{k})\delta_{h}(\omega(\mathbf{k})-k_{\parallel}v_{\parallel}-l\omega_{c}/\gamma(\mathbf{p}))\frac{p}{E_{h}}\right]dp_{\perp}dp_{\parallel}.
    \end{equation*}

    Therefore inequality (\ref{L1bound}) still holds.
\end{itemize}

Combine inequalities (\ref{Linfbound}), (\ref{discretebound}), and (\ref{L1bound}) to obtain
\begin{equation*}
    \left\vert\int_{\Omega_{p}^{L}}d\mathbf{p}\{\mathcal{L}_{h}E_{h}\mathcal{L}_{h}f_{h}^{s+1}\mathcal{B}_{h}\}\right\vert\leq 2 \pi \frac{C_{1}\cdot C_{2}\cdot C_{3}\cdot C_{4}}{h_{p}}\left\Vert f_{h}^{0}\right\Vert_{L^{2}(\Omega_{p}^{L})},
\end{equation*}
which enables us to define the uniform-in-time upper bound,
\begin{equation*}
    \Delta t_{M}\coloneqq  \frac{(1-\epsilon)h_{p}}{2 \pi 
 \cdot C_{1}\cdot C_{2}\cdot C_{3}\cdot C_{4}\left\Vert f_{h}^{0}\right\Vert_{L^{2}(\Omega_{p}^{L})}}.
\end{equation*}

It can be easily verified that any $\Delta t < \Delta t_{M}$ satisfies condition (\ref{stepsize}).
\end{proof}

The condition does not need to be calculated explicitly, because we can adapt time step size a posteriori in the code: monitor the cell averages, if any cell average of the temporary solution $W_{h}^{s+1,*}$ is non-positive, replace $\Delta t$ with $0.5 \Delta t$ and calculate $W_{h}^{s+1,*}$ again.

\bigskip
Now let us discuss the effect of slope limiters on conservation laws. If $\alpha_{2}=0$, there is no need for any slope limiter. If $\alpha_{2}=1$, we apply the slope limiter $\theta$ and obtain $W_{h}^{s+1}=\theta(W^{s+1,*}_{h})$. According to Zhang et al.\cite{zhang2010maximum}, the cell average is preserved, i.e. $\int_{R_{k}^{j}}W_{h}^{s+1,*}d\mathbf{k} = \int_{R_{k}^{j}}W_{h}^{s+1}d\mathbf{k}$. In other words, $\left(W_{h}^{s+1,*}, \eta \right)_{k} = \left(W_{h}^{s+1}, \eta \right)_{k}$, for any piecewise constant test function, i.e. $\forall \eta \in \mathcal{W}_{h}^{0}$. Therefore, to preserve discrete conservation laws, in the definition of the discrete directional differential operator $\mathcal{L}_{h}$, we need to pick a projection $\Pi_{k,h}$ such that $\Pi_{k,h} U$ belongs to $\mathcal{W}_{h}^{0} \subset \mathcal{W}_{h}^{\alpha_2}$ for any function $U$.

\bigskip

\section{Numerical Results} \label{results}
\subsection{Problem Setting}

Although the emission/absorption kernel contains a summation from $l=-\infty$ to $l=+\infty$, it is not practical to perform that numerically. In practice, we keep the dominant part of those terms. In the following example, we will only consider one term with $l=1$, associated with the anomalous Doppler resonance. We used the dispersion relation $\omega(\mathbf{k})$ of the whistler wave in cold magnetized plasma(see Appendix), with electron gyro-frequency $\omega_{c} = -2 \omega_{p}$. 

Set the cut-off computational domain as follows,
\begin{equation*}
    \begin{split}
        \Omega^{L}_p&=\{(p_{\parallel},p_{\perp}):p_{\parallel}\in(-5mc,25mc),\ p_{\perp}\in(0,15mc)\}\\
        \Omega^{L}_k&=\{(k_{\parallel},k_{\perp}):k_{\parallel}\in(0.05\frac{\omega_{p}}{c},0.65\frac{\omega_{p}}{c}),\ k_{\perp}\in(0,0.6\frac{\omega_{p}}{c})\}
    \end{split}
\end{equation*}

Take piecewise linear quadrilateral basis $\mathcal{G}_{h}^{1}=\{f(p_{\parallel},p_{\perp}) \in C^0(\Omega^{L}_p):f|_{R_{p}}\in Q^{1}(R_{p}),\forall~R_{p}\in\mathcal{T}_{h}^{p}\}$ and piecewise constant basis $\mathcal{W}_{h}^{0}=\{W(k_{\parallel},k_{\perp}):W|_{R_{k}}\in Q^{0}(R_{k}),\forall~R_{k}\in\mathcal{T}_{h}^{k}\}$ as our test spaces. Choose the $L^2$ orthogonal projections $\Pi_{p,h}$ and $\Pi_{k,h}$ as stated in Corollary \ref{orthog}.

The numerical experiment is performed with $75\times75$ elements in $\Omega_{p}^{L}$, and $40\times40$ elements in $\Omega_{k}^{L}$. The initial time step size is set as $\Delta t = 1.0\times10^{4} \frac{1}{2\pi \omega_{p}}$.

The integration on resonance manifold is performed with Gauss-Legendre quadrature on $\Omega^{L}_{p}$ and the marching simplex method\cite{gueziec1995exploiting,min2007geometric} on $\Omega^{L}_{k}$.

\bigskip
Consider the following initial conditions, which is the so-called 'bump on tail instability' configuration. 
\begin{equation*}
    \begin{cases}
        \begin{aligned}
            f(p_{\parallel},p_{\perp})|_{t=0}
            &= \left[10^{-5}\frac{1}{\sqrt{\pi}}\exp\left(-\left(\frac{p_{\parallel}}{mc}-20\right)^{2}-\left(\frac{p_{\perp}}{mc}\right)^2\right) \right] \frac{n_{0}}{m^{3}c^{3}}\\
            W(k_{\parallel},k_{\perp})|_{t=0}
            &=10^{-5}\frac{n_{0}mc^{2}}{(\omega_{p}/c)^{3}}
        \end{aligned}
    \end{cases}
\end{equation*}

\begin{remark}
The bump on tail configuration actually refers to the sum of a bulk and a bump, i.e. $f(\mathbf{p},t) = f_{c}(\mathbf{p}) + f_{b}(\mathbf{p},t)$, where the cold bulk $f_{c}(\mathbf{p}) \approx \frac{n_{0}}{m^{3}c^{3}}\delta(\mathbf{p})$, and the bump $f_b(\mathbf{p},t)$ is a peak with a much smaller population, centered far from the origin. However, as shown in the following equation, 
\begin{equation*}
    \begin{split}
        &\frac{\partial(f-f_{c})}{\partial t}=\mathbb{B}(W,f)=\mathbb{B}(W,f-f_{c})+\mathbb{B}(W,f_{c})=\mathbb{B}(W,f-f_{c}),\\
        &\frac{\partial W}{\partial t}=\mathbb{H}(f,W)=\mathbb{H}(f-f_{c},W)+\mathbb{H}(f_{c},W)=\mathbb{H}(f-f_{c},W),
    \end{split}
\end{equation*}
we do not have to really compute the contribution from $f_{c}$. 

\end{remark}

\subsection{Temporal Evolution}
In analogous to the analysis done by Kennel and Engelmann \cite{kennel_velocity_1966}, for a given wave vector $\mathbf{k}$, the characteristics associated with directional differential operator $\mathcal{L}$ is 

\begin{equation}
    z(p_{\parallel},p_{\perp})\coloneqq\frac{\omega}{k_{\parallel}}p_{\parallel}-E(p_{\parallel},p_{\perp})=\frac{\omega}{k_{\parallel}}p_{\parallel}-\sqrt{m^{2}c^{4}+p_{\parallel}^{2}c^{2}+p_{\perp}^{2}c^{2}}=\text{const},
    \label{diffcontour}
\end{equation}
which is the isoenergy contour in the reference frame moving at the wave's phase velocity.

When the wave \textit{sed} $W(\mathbf{k},t)$ is concentrated around the given
$\mathbf{k}$, the contours as illustrated in Figure(\ref{charac}) indicates the principal diffusion direction. For the specific problem setting, $\frac{\omega}{k_{\parallel}}$
is small, hence the contour lines are almost concentric circles.

In Figure(\ref{fWevol}) we show the evolution of electron \textit{pdf} $f(p_{\parallel},p_{\perp},t)$ and wave \textit{sed} $W(k_{\parallel},k_{\perp},t)$. It can be observed that the bump on tail results in the excitation of the approximate waves in a narrow region of spectral space $\Omega^{L}_{k}$, and as predicted by Equation(\ref{diffcontour}), the whistler waves in turn cause anisotropic diffusion of electron \textit{pdf} almost along the contour lines in Figure(\ref{charac}).

\begin{figure}[htbp!]
    \centering
    \begin{subfigure}[b]{\textwidth}
        \centering
        \includegraphics[width=\textwidth]{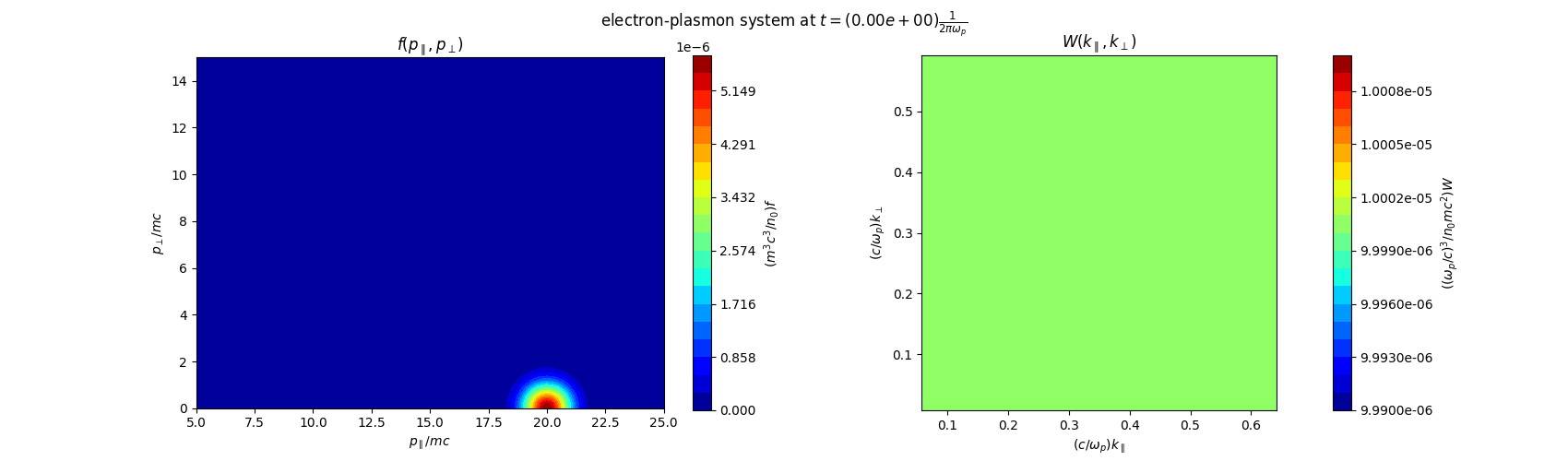}
    \end{subfigure}
    \begin{subfigure}[b]{\textwidth}
        \centering
        \includegraphics[width=\textwidth]{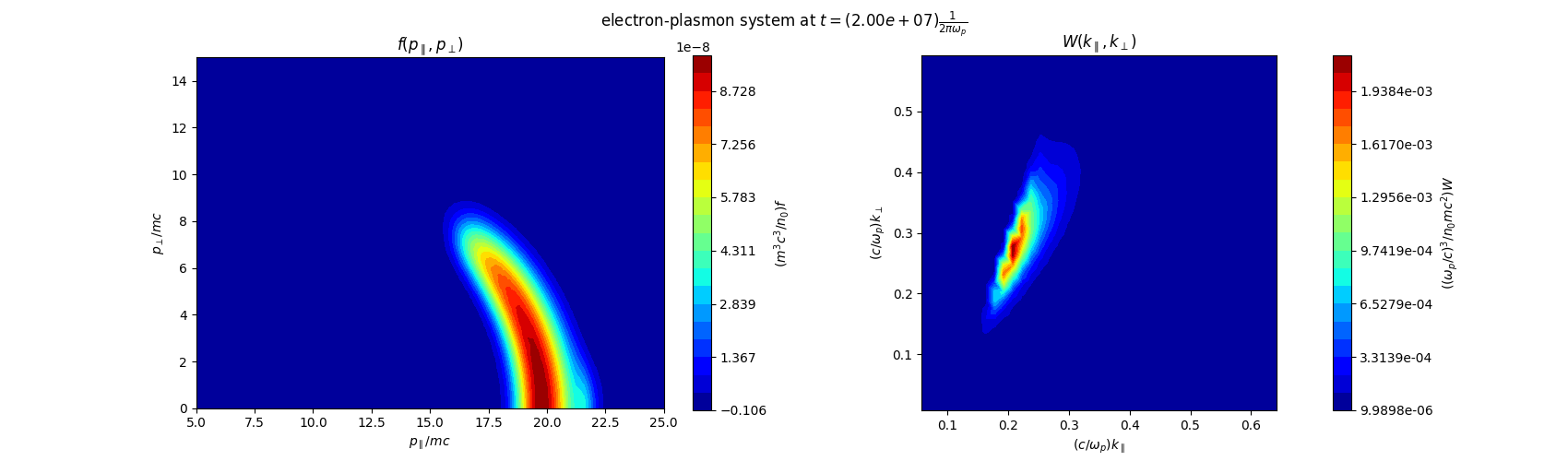}
    \end{subfigure}
    \begin{subfigure}[b]{\textwidth}
        \centering
        \includegraphics[width=\textwidth]{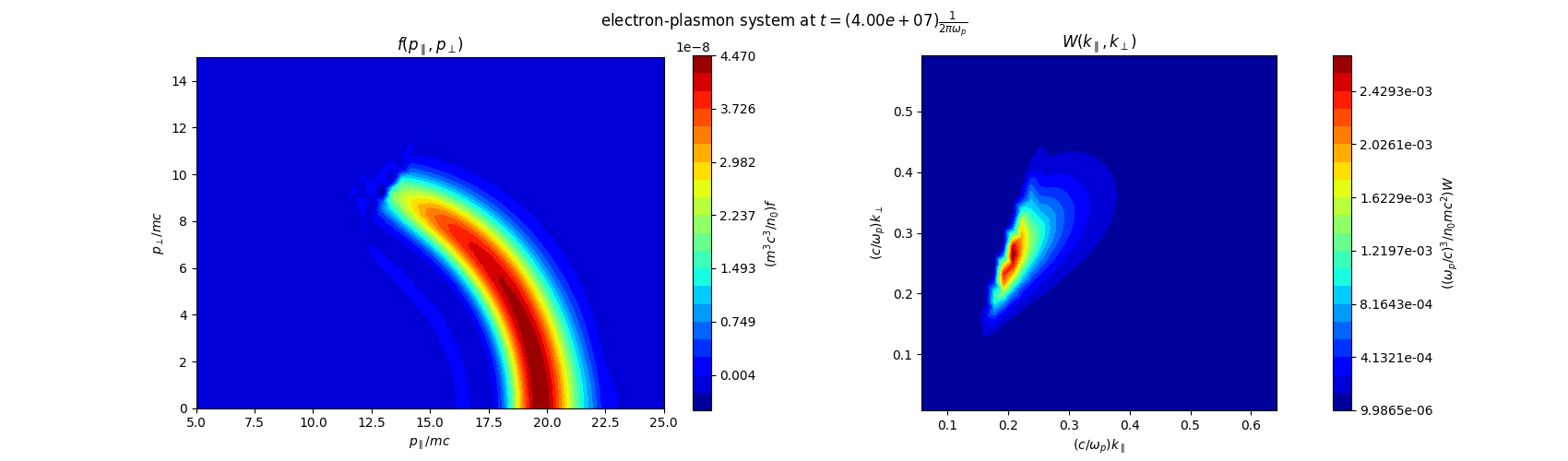}
    \end{subfigure}
    \caption{Temporal evolution of the electron \textit{pdf} and wave \textit{sed}.}
    \label{fWevol}
\end{figure}

\begin{figure}[htbp!]
    \centering
    \includegraphics[width=0.8\textwidth]{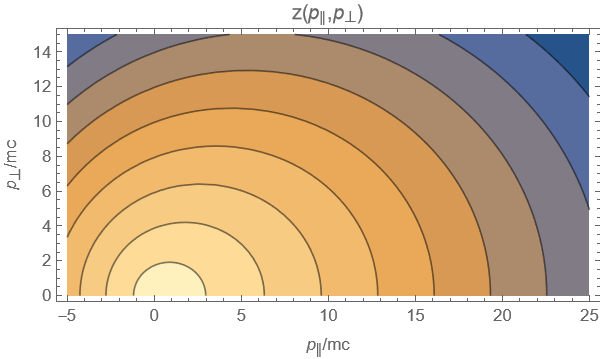}
    \caption{The characteristics of directional differential operator $\mathcal{L}$ given $(k_{\parallel},k_{\perp})=(0.2, 0.3)$.}
    \label{charac}
\end{figure}

\subsection{Verification of Conservation}

To verify the discrete conservation property of the proposed scheme, we define the relative error for conserved quantity as follows,
\begin{equation*}
    e_{rel}\left(\mathcal{Q}_{tot,h}\right)\coloneqq\frac{\Vert\mathcal{Q}_{tot,h}-\mathcal{Q}_{tot,h}^{0}\Vert_{L^\infty(0,T_{max})}}{\mathcal{Q}_{tot,h}^{0}},
\end{equation*}
where $\mathcal{Q}$ is the conserved quantities defined in Theorem \ref{uncond}.

Then with $T_{max} = 8.0\times 10^{7} \frac{1}{2\pi \omega_{p}}$, we have
\begin{equation*}
    \begin{split}
        e_{rel}(\mathcal{M}_{tot,h})&=4.96\times 10^{-14},\\
        e_{rel}(\mathcal{P}_{\parallel,tot,h})&=4.58\times 10^{-14},\\
        e_{rel}(\mathcal{E}_{tot,h})&= 4.69 \times 10^{-14}.
    \end{split}
\end{equation*}

For the evolution of the electron-plasmon system momentum and energy, see solid lines in Figure(\ref{transfer}).

\subsection{Comparison of Different Dispersion Relations}
The above results were obtained with the exact whistler wave dispersion relation $\omega_{\text{imp}}(\mathbf{k})$ for cold magnetized plasma, given implicitly by Equation(\ref{DispRelEq}). One might wonder what if we replace it with a simpler explicit approximate relation, for instance,
\begin{equation*}
    \omega_{\text{exp}}(\mathbf{k})=|\omega_{c}|\frac{|k_{\parallel}|kc^{2}}{\omega^{2}_{p}\sqrt{1+k^{2}c^{2}\omega_{c}^{2}/\omega_{p}^{4}}}
\end{equation*}
which is asymptotic to the implicit relation $\omega_{\text{imp}}(\mathbf{k})$ when $k$ is small, i.e.
\begin{equation*}
    \lim_{k\rightarrow0}\left(\omega_{\text{exp}}(\mathbf{k})-\omega_{\text{imp}}(\mathbf{k})\right)=0.
\end{equation*}

 \begin{figure}[htbp!]
    \centering
    \begin{subfigure}[b]{0.45\textwidth}
        \centering
        \includegraphics[width=\textwidth]{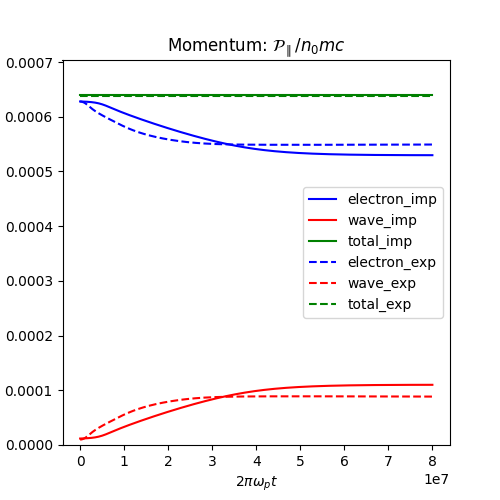}
        \caption{Momentum evolution.}
    \end{subfigure}
    \begin{subfigure}[b]{0.45\textwidth}
        \centering
        \includegraphics[width=\textwidth]{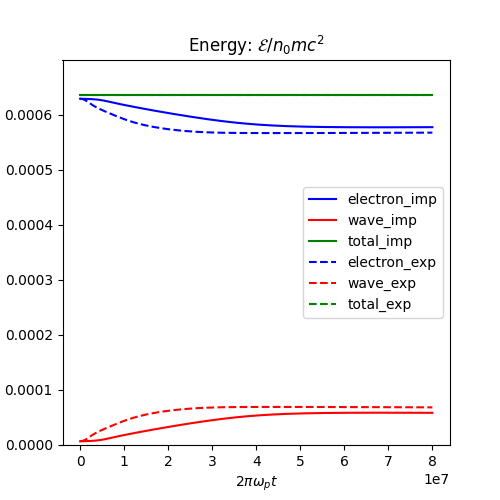}
        \caption{Energy evolution.}
    \end{subfigure}
    \caption{Comparison between $\omega_{\text{exact}}$ and $\omega_{\text{approx}}$}
    \label{transfer}
\end{figure}

As shown in Figure(\ref{transfer}), for both cases, energy and momentum are transferred from particles to waves. Meanwhile, we do observe a different transfer rate for the approximate whistler dispersion relation when compared to the exact implicit dispersion relation derived from Equation(\ref{DispRelEq}) in the Appendix.

\bigskip

\section{Summary}
We studied the numerical method for the initial value problem associated with the relativistic quasilinear diffusion model in magnetized plasma. We showed that a conservative semi-discrete form can be derived by adopting a novel integro-differential form of this wave-particle interaction system. We incorporated the marching simplex algorithm in numerical integration on the resonance manifold. A semi-implicit time discretization was introduced to ensure the stability of the particle \textit{pdf} and positivity of wave \textit{sed}, which also preserves conservation in the fully discrete form. In the end, we presented our numerical results for the bump-on-tail instability, and the conservation properties are verified. 

In the future, we will consider the problem with the spatial non-uniform setting, and other factors such as Landau collision operator and external electric field will be included. Error estimates for the Galerkin scheme will also be investigated.
\bigskip

\section{Appendix : Waves in Cold Magnetized Plasma}
The directional differential operator and the emission/absorption kernel both vary for different wave modes. Since our numerical experiment is based on the whistler mode, we introduce the wave modes in cold magnetized plasma here to make the paper self-contained.

Inside the medium with conductivity tensor $\sigma$, we have the linear relation between current and field $\mathbf{J}=\sigma \cdot \mathbf{E}$. Apply the Fourier transform to the Maxwell equations, we obtain that 
\begin{equation*}
    \mathbf{k}\times\mathbf{k}\times\hat{\mathbf{E}}=\mu_{0}\varepsilon_{0}\omega^{2}(\frac{i\mathbf{\sigma}}{\omega\varepsilon_{0}}+I)\cdot\hat{\mathbf{E}}
\end{equation*}
Write it with Einstein's summation notation, we have the wave equation in spectral form,
\begin{equation}
    (k_{\alpha}k_{\beta}-\delta_{\alpha\beta}k^{2}+\frac{\omega^2}{c^2}\varepsilon_{\alpha\beta})\hat{E}_{\beta}=0, \beta=1,2,3 
    \label{maxwell}
\end{equation}
where the dielectric tensor $\varepsilon_{\alpha\beta}$ is dependent on $\sigma$.

Let $M(\omega,\mathbf{k})=k_{\alpha}k_{\beta}-\delta_{\alpha\beta}k^{2}+\frac{\omega^2}{c^2}\varepsilon_{\alpha\beta}$, then the wave equation(\ref{maxwell}) has nontrivial solution if and only if

\begin{equation}
    Det[M(\omega,\mathbf{k})]=0
    \label{DispRelEq}
\end{equation}

Equation(\ref{DispRelEq}) gives the graph of implicit function $\omega(\mathbf{k})$, which is known as the dispersion relation. 

\bigskip
The above discussion works for any medium. Now we focus on the plasma. Consider cold magnetized plasma with background field $\mathbf{B}_{0}(\mathbf{x})=B_0 \mathbf{b}$, where $B_0$ is constant and $\mathbf{b}$ is a fixed unit vector. A plasma is "cold" when the waves propagate faster than its thermal speed. 

The following dielectric tensor for cold magnetized plasmas can be found in textbooks\cite{stix1992waves,thorne2017modern}. (For simplicity, ion motion is neglected here.)
\begin{equation*}
    \varepsilon_{\alpha \beta}(\omega) \equiv \varepsilon \delta_{\alpha \beta}+i g e_{\alpha \beta \gamma} b_{\gamma}+(\eta-\varepsilon) b_{\alpha} b_{\beta},
\end{equation*}
where
\begin{equation}
    \begin{split}
        \varepsilon
        &=\varepsilon^{H}\equiv1-\frac{\omega_{pe}^{2}}{\omega^{2}-\omega_{ce}^{2}},\\
        g
        &=g^{H}\equiv-\frac{\omega_{ce}}{\omega}\frac{\omega_{pe}^{2}}{\omega^{2}-\omega_{ce}^{2}},\\
        \eta
        &=\eta^{H}\equiv1-\frac{\omega_{pe}^{2}}{\omega^{2}}.
    \end{split}
    \label{dielec_compo}
\end{equation}

% \begin{equation*}
%     \begin{split}
%         \varepsilon
%         &=\varepsilon^{H}\equiv1-\frac{\omega_{pe}^{2}}{\omega^{2}-\omega_{ce}^{2}}-\frac{\omega_{pi}^{2}}{\omega^{2}-\omega_{ci}^{2}},\\
%         g
%         &=g^{H}\equiv-\frac{\omega_{ce}}{\omega}\frac{\omega_{pe}^{2}}{\omega^{2}-\omega_{ce}^{2}}-\frac{\omega_{ci}}{\omega}\frac{\omega_{pi}^{2}}{\omega^{2}-\omega_{ci}^{2}},\\
%         \eta
%         &=\eta^{H}\equiv1-\frac{\omega_{pe}^{2}}{\omega^{2}}-\frac{\omega_{pi}^{2}}{\omega^{2}}.
%     \end{split}
% \end{equation*}

There are two parameters in the above formulas. The electron gyro-frequency $\omega_{ce}=-|e|B/m_{e}c$ is proportional to the background magnetic field. The plasma frequency $\omega_{pe}=\sqrt{4\pi n_{e} e^2/m_{e} }$ is proportional to the square root of particle density. 

\bigskip
Recall that the dispersion relation $\omega(\mathbf{k})$ is given implicitly in Equation(\ref{DispRelEq}), one might wonder whether $\omega(\mathbf{k})$ is multi-valued, and for a specific branch, whether it is well-defined for any $\mathbf{k} \in \Omega_{k}$. Textbooks never elaborate on this issue, therefore we provide an answer here.

Define the parallel component and perpendicular component of wave vector $\mathbf{k}$, $k_{\parallel} \coloneqq \mathbf{k}\cdot\mathbf{B}/|\mathbf{B}|$, $k_{\perp} \coloneqq \sqrt{\mathbf{k}^2-k_{\parallel}^2}$. Denote the magnitude $\mathbf{k}$ as $k$, and define the cosine of polar angle as $\xi \coloneqq {k_{\parallel}/k}$. If ion motion is neglected, then we have what follows.

\begin{prop}
$\forall~ (k_{\parallel},k_{\perp}) \in (0,+\infty)\times(0,+\infty)$, $\exists~ 0<\omega_1<\omega_2<\omega_3<\omega_4<\infty$ s.t. 
\begin{equation*}
    Det[M(\omega_j,k_{\parallel},k_{\perp})]=0,~ j=1,2,3,4
\end{equation*}
i.e. the equation admits exactly $4$ positive single-value implicit functions $\omega_j(k_{\parallel},k_{\perp})$, $j=1,2,3,4$, on domain $(0,+\infty)\times(0,+\infty)$, moreover, $\omega_{j}(k,\xi)\coloneqq\omega_{j}(k_{\parallel}(k,\xi),k_{\perp}(k,\xi))$ satisfy that  $\frac{\partial}{\partial k}\omega_{j}(k,\xi)\geq0,~\forall~\xi \in (0,1)$.
\end{prop}

The first branch $\omega_{1}(\mathbf{k})$ is defined on the whole spectral space $\mathbb{R}^{3}_{k}$. In a relatively strong magnetic field($\omega_{p}^{2}/\omega_{c}^{2}\ll1$), the whistler wave actually refers to waves with wave vector $\mathbf{k} \in \Omega_{w} \subsetneq \mathbb{R}^{3}_{k}$, and meanwhile has frequency $\omega_{1}(\mathbf{k})$. Outside the region $\Omega_{w}$, the first branch has another name. For details, see Aleynikov et al.\cite{aleynikov_stability_2015}. In our numerical experiment, the cut-off domain $\Omega^{L}_{k} \subset \Omega_{w}$, therefore we say that we consider the whistler wave. 

The polarization vector components in the emission/absorption kernel are given below
\begin{equation}
    \begin{split}
        E_{1}(\mathbf{k})&=1,\\
        E_{2}(\mathbf{k})&=i\frac{g}{\varepsilon-N^{2}},\\
        E_{3}(\mathbf{k})&=-\frac{N_{\parallel}N_{\perp}}{\eta-N_{\perp}^{2}},
    \end{split}
    \label{polarize}
\end{equation}
where $\mathbf{N}=\frac{\mathbf{k}}{\omega}$ is the refractive index.

\bigskip

{ \bf Acknowledgements.}
The authors thank and gratefully
acknowledge the support from the Oden Institute of Computational
Engineering and Sciences and the University of Texas Austin. This project 
was supported by  funding from NSF DMS: 2009736 and DOE DE-SC0016283
project Simulation Center for Runaway Electron Avoidance and Mitigation. 
\bigskip
%\section{Acknowledgements}
%\subfile{sections/Acknowledgements}
\newpage
\bibliographystyle{plain}
\bibliography{main.bib}

\begin{thebibliography}{10}

\bibitem{aleynikov_stability_2015}
Pavel Aleynikov and Boris Breizman.
\newblock Stability analysis of runaway-driven waves in a tokamak.
\newblock {\em Nuclear Fusion}, 55(4):043014, April 2015.

\bibitem{bardos2021diffusion}
Claude Bardos and Nicolas Besse.
\newblock Diffusion limit of the vlasov equation in the weak turbulent regime.
\newblock {\em Journal of Mathematical Physics}, 62(10):101505, 2021.

\bibitem{besse2011validity}
Nicolas Besse, Yves Elskens, DF~Escande, and Pierre Bertrand.
\newblock Validity of quasilinear theory: refutations and new numerical confirmation.
\newblock {\em Plasma Physics and Controlled Fusion}, 53(2):025012, 2011.

\bibitem{breizman_physics_2019}
Boris~N. Breizman, Pavel Aleynikov, Eric~M. Hollmann, and Michael Lehnen.
\newblock Physics of runaway electrons in tokamaks.
\newblock {\em Nuclear Fusion}, 59(8):083001, August 2019.

\bibitem{doi1991efficient}
Akio Doi and Akio Koide.
\newblock An efficient method of triangulating equi-valued surfaces by using tetrahedral cells.
\newblock {\em IEICE TRANSACTIONS on Information and Systems}, 74(1):214--224, 1991.

\bibitem{drummond1962non}
WE~Drummond and D~Pines.
\newblock Non-linear stability of plasma oscillations.
\newblock 1962.

\bibitem{gueziec1995exploiting}
Andr{\'e} Gu{\'e}ziec and Robert Hummel.
\newblock Exploiting triangulated surface extraction using tetrahedral decomposition.
\newblock {\em IEEE Transactions on visualization and computer graphics}, 1(4):328--342, 1995.

\bibitem{huang2023existence}
Kun Huang and Irene~M. Gamba.
\newblock Existence of global weak solutions to quasilinear theory for electrostatic plasmas, 2023.

\bibitem{kaufman1971resonant}
Allan~N Kaufman.
\newblock Resonant interactions between particles and normal modes in a cylindrical plasma.
\newblock {\em The Physics of Fluids}, 14(2):387--397, 1971.

\bibitem{kennel_velocity_1966}
C.~F. Kennel.
\newblock Velocity {Space} {Diffusion} from {Weak} {Plasma} {Turbulence} in a {Magnetic} {Field}.
\newblock {\em Physics of Fluids}, 9(12):2377, 1966.

\bibitem{lerche1968quasilinear}
I~Lerche.
\newblock Quasilinear theory of resonant diffusion in a magneto-active, relativistic plasma.
\newblock {\em The Physics of Fluids}, 11(8):1720--1727, 1968.

\bibitem{min2007geometric}
Chohong Min and Fr{\'e}d{\'e}ric Gibou.
\newblock Geometric integration over irregular domains with application to level-set methods.
\newblock {\em Journal of Computational Physics}, 226(2):1432--1443, 2007.

\bibitem{nitsche1971variationsprinzip}
Joachim Nitsche.
\newblock {\"U}ber ein variationsprinzip zur l{\"o}sung von dirichlet-problemen bei verwendung von teilr{\"a}umen, die keinen randbedingungen unterworfen sind.
\newblock In {\em Abhandlungen aus dem mathematischen Seminar der Universit{\"a}t Hamburg}, volume~36, pages 9--15. Springer, 1971.

\bibitem{rosenbluth1997theory}
MN~Rosenbluth and SV~Putvinski.
\newblock Theory for avalanche of runaway electrons in tokamaks.
\newblock {\em Nuclear fusion}, 37(10):1355, 1997.

\bibitem{shapiro1962nonlinear}
VD~Shapiro and VI~Shevchenko.
\newblock On the nonlinear theory of interaction between charged particle beams and a plasma in a magnetic field.
\newblock {\em Zhur. Eksptl'. i Teoret. Fiz.}, 42, 1962.

\bibitem{stix1992waves}
Thomas~H Stix.
\newblock {\em Waves in plasmas}.
\newblock Springer Science \& Business Media, 1992.

\bibitem{thorne2017modern}
Kip~S Thorne and Roger~D Blandford.
\newblock {\em Modern classical physics: optics, fluids, plasmas, elasticity, relativity, and statistical physics}.
\newblock Princeton University Press, 2017.

\bibitem{vedenov1967theory}
AA~Vedenov.
\newblock Theory of a weakly turbulent plasma.
\newblock {\em Reviews of plasma physics}, pages 229--276, 1967.

\bibitem{vedenov1961nonlinear}
AA~Vedenov, EP~Velikhov, and RZ~Sagdeev.
\newblock Nonlinear oscillations of rarified plasma.
\newblock {\em Nuclear Fusion}, 1(2):82, 1961.

\bibitem{zhang2017conservative}
Chenglong Zhang and Irene~M Gamba.
\newblock A conservative scheme for vlasov poisson landau modeling collisional plasmas.
\newblock {\em Journal of Computational Physics}, 340:470--497, 2017.

\bibitem{zhang2018conservative}
Chenglong Zhang and Irene~M Gamba.
\newblock A conservative discontinuous galerkin solver for the space homogeneous boltzmann equation for binary interactions.
\newblock {\em SIAM Journal on Numerical Analysis}, 56(5):3040--3070, 2018.

\bibitem{zhang2010maximum}
Xiangxiong Zhang and Chi-Wang Shu.
\newblock On maximum-principle-satisfying high order schemes for scalar conservation laws.
\newblock {\em Journal of Computational Physics}, 229(9):3091--3120, 2010.

\end{thebibliography}

\end{document}